\documentclass[conference]{IEEEtran}
\IEEEoverridecommandlockouts
\usepackage{cite}
\usepackage{color}
\usepackage{url}
\usepackage{hyperref}
\usepackage{amsmath,amssymb,amsfonts}
\usepackage{amsthm}

\newtheorem{lemma}{Lemma}
\newtheorem{corollary}{Corollary}
\newtheorem{assumption}{Assumption}
\newtheorem{remark}{Remark}

\usepackage{amsmath}

\usepackage{graphicx}
\usepackage{xcolor}
\usepackage{textcomp}
\usepackage{makecell}
\usepackage{xcolor}
\usepackage{multirow}
\usepackage{cuted}

\usepackage{silence}
\usepackage{dblfloatfix}
\usepackage{algorithm}
\usepackage{comment}
\usepackage{array}
\usepackage{tabularx}
\usepackage{booktabs}
\usepackage{balance}
\usepackage{array}    % needed for the ragged-right paragraph column below
\usepackage{booktabs} % for \toprule \midrule \bottomrule (if not already loaded)

\newcolumntype{L}[1]{>{\raggedright\arraybackslash}p{#1}}

\usepackage{cuted}
\newtheorem{proposition}{Proposition}

\usepackage{booktabs}

\usepackage{algpseudocode} 
\usepackage{setspace}
\usepackage[caption=false,font=footnotesize]{subfig}

\usepackage{footnote}
\makesavenoteenv{figure*}

\newtheorem{theorem}{Theorem}
\newtheorem{definition}{Definition}
\newtheorem{example}{Example}

\def\BibTeX{{\rm B\kern-.05em{\sc i\kern-.025em b}\kern-.08em
    T\kern-.1667em\lower.7ex\hbox{E}\kern-.125emX}}

\begin{document}

\title{From Propagation to Protection: Risk-Aware Diffusion for Harm Minimization in Signed Social Networks
}

\author{
\IEEEauthorblockN{
Aaqib Zahoor,
Janibul Bashir,
and Iqra Altaf Gillani
}
\IEEEauthorblockA{
Department of Information Technology,
National Institute of Technology Srinagar, India\\
Email: aaqib\_phaite003@nitsri.ac.in,
iqraaltaf@nitsri.ac.in
}
}

\maketitle

\begin{abstract}
Real-world social relationships are not uniformly supportive. Information through hostile connections can increase resistance, anxiety, or misinformation rather than adoption. Classical models such as Independent Cascade and Linear Threshold, together with Influence Maximization (IM), which seeks to maximize spread from a limited seed set, treat activation as discrete and irreversible. Its counterpart, Influence Minimization (Inf-Min), seeks to limit undesirable spread but similarly relies on simplified activation assumptions. Signed extensions incorporate polarity but largely retain this irreversibility, leaving no room for beliefs to weaken, reverse, or recover under competing influence. Moreover, both objectives typically treat individuals uniformly, without accounting for differences in vulnerability or prioritizing the protection of those most at risk.
We introduce \textit{RASH}, a signed, susceptibility-aware diffusion model in which node awareness is continuous, bounded, and non-monotonic, and prove that despite this added expressiveness it remains monotone and $\gamma$-weakly submodular in the regimes where only positive or negative edges exist, preserving tractable greedy approximation guarantees where strict submodularity provably fails. Building on RASH, we formulate \textit{Harm Minimization (HM)}, which maximizes the aggregate reach while minimizing the awareness shortfall (harm). We prove HM is NP-hard, yet its harm-reduction formulation inherits the same monotonicity and weak-submodularity structure, admitting a greedy algorithm that carries a bounded approximation ratio. Across six structurally diverse signed networks, RASH is the only diffusion model tested to the best of our knowledge that ever allows
awareness to reverse after activation, allowing sustained discouraging
influence to drive awareness from positive toward negative and HM achieves the highest harm reduction of any method evaluated including its own boundary cases (IM and Inf-Min).
\end{abstract}

\begin{IEEEkeywords}
Temporal Networks, Information Diffusion , Signed Social Networks, Influence Maximization  
\end{IEEEkeywords}

\section{Introduction}
\label{sec:Intro}

Social networks provide the pathways through which opinions, behaviors, and information spread among individuals. Since phenomena such as product adoption, political mobilization, public-health campaigns, and misinformation all evolve over these networks, understanding how influence propagates has become a fundamental problem in computer science, sociology, and public policy. However, real social relationships are not uniformly supportive. They can reinforce or discourage beliefs and behaviors, motivating the use of signed social networks, where edges represent either positive or negative influence. This distinction fundamentally changes the nature of diffusion. In unsigned networks, reaching a node is always beneficial, whereas in signed networks, influence transmitted through negative relationships may increase resistance, anxiety, or misinformation instead of promoting adoption.

Classical information diffusion models like Independent Cascade (IC) and Linear Threshold (LT)~\cite{zahoor2025diffusionmodelsinfluencemaximization} were largely developed before this distinction mattered to the problems they were built to solve. Such models, together with the influence maximization framework built upon them, assume purely positive influence and irreversible activation: a node is either inactive or permanently active, and every activation is, by construction, a success~\cite{4a}. This assumption makes the models analytically tractable but leaves no room for a node's belief to weaken, reverse, or degrade under competing influence. Extensions~\cite{ICE,LTE,PID,PLID} that introduce edge signs partially address this limitation, but nearly all retain irreversible activation: once a node commits to a sign, its state is permanent. This leaves existing models unable to capture phenomena such as relapse, reversal of radicalization, and waning conviction that frequently occur in real signed social networks.

This gap is not merely theoretical. Consider a mental-health awareness campaign, a cybersecurity advisory, a safety recall, or a crisis alert circulating on a social platform. Individuals in these networks do not simply receive information or not; rather, they receive competing information, and their eventual state depends on which influence---supportive or discouraging dominates, and on how resistant they are to negative pressure. Modeling this requires three ingredients absent from classical diffusion models: (i) a signed representation of edges distinguishing reinforcing from discouraging relationships, (ii) a notion of individual susceptibility capturing why the same discouraging signal affects different people differently, and (iii) a dynamic that allows awareness to evolve in both directions over time, including reversal after prior activation.

We introduce the \emph{Risk-Aware Signed Heterogeneous (RASH)} diffusion model to provide these three capabilities. The underlying picture is simple: a node weighs the support it receives from friendly neighbors against the discouragement it receives from hostile ones, in proportion to how susceptible that node is to negative pressure, and it only updates its belief when this combined signal is strong enough to clear a resistance threshold. Concretely, each node maintains a continuous awareness state in $[-1,1]$ that aggregates supportive and discouraging influence through a susceptibility-weighted combination, passes the result through a bounded nonlinear activation, and updates its state through this threshold-gated rule while retaining memory of its prior belief. Unlike classical models and their signed extensions, RASH consequently permits non-monotonic trajectories: a node's awareness can grow, shrink, or reverse sign entirely as the balance of influence in its neighborhood changes. We show that despite this added expressiveness, RASH retains enough mathematical structure---monotonicity and $\gamma$-weak submodularity---to support greedy optimization with provable approximation guarantees in both the pure-positive (influence maximization) and pure-negative (influence minimization) regimes.

RASH's ability to represent competing, reversible influence forms the foundation of our framework, but it also exposes a deeper question: what should a diffusion campaign actually optimize? Classical influence maximization seeks the seed set that reaches the largest number of individuals~\cite{4a}, implicitly treating every unreached individual as equally costly to miss. Its counterpart, influence minimization, instead seeks a blocking set that suppresses harmful diffusion, but does so uniformly, without distinguishing whose protection matters most. Figure~\ref{fig:comparisonofInfluenceMax} makes this concrete: each individual carries a personal safety threshold and a vulnerability weight that vary across the four representative campaigns shown, yet under both IM and Inf-Min every node's awareness is fixed for good the moment it is determined, and no node is treated as more urgent to protect than any other. Both objectives therefore share a structural blind spot---one that matters precisely for the protective campaigns motivating this work, leaving open the question of how a campaign should allocate its seed budget when some individuals are far more vulnerable, and far more reachable by harm, than others.
\begin{figure}[h]
    \centering
     \caption{Comparison of Influence Maximization (IM), Influence Minimization (Inf-Min), and the proposed Harm Minimization (HM) in signed social networks. Unlike IM and Inf-Min, HM explicitly protects vulnerable nodes by limiting the propagation of harmful awareness while preserving beneficial diffusion}
    \label{fig:comparisonofInfluenceMax}
    \includegraphics[width=0.99\linewidth]{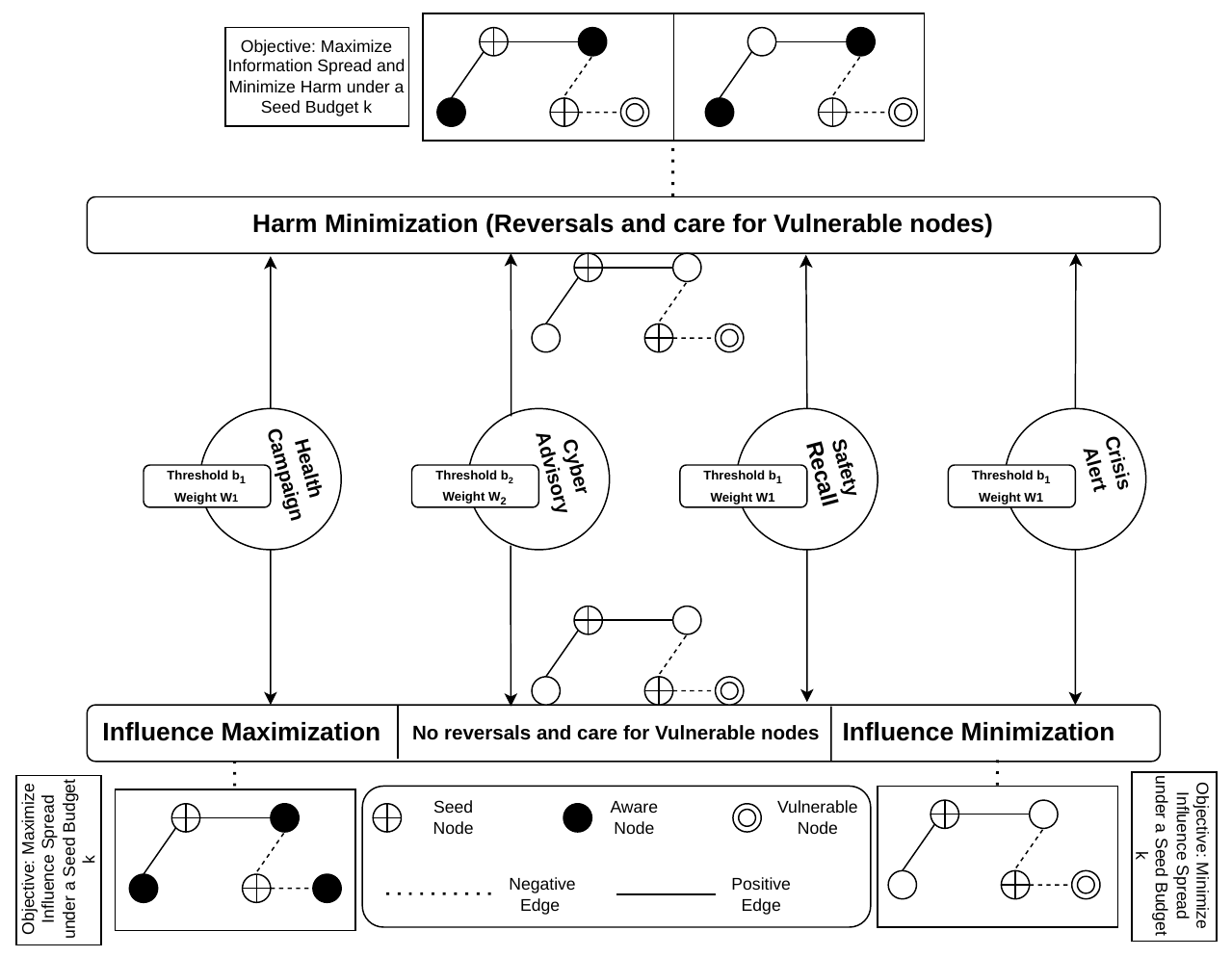}
   \end{figure}

This is the gap our second contribution addresses. We formulate the \emph{Harm Minimization (HM)} objective, which realizes exactly the behavior contrasted with IM and Inf-Min in Fig.~\ref{fig:comparisonofInfluenceMax}: awareness that can both recede and recover, and seed selection that explicitly prioritizes individuals most susceptible to harm rather than merely maximizing total reach. Each individual is assigned a personal safety threshold, representing the minimum awareness required to be considered protected, and a vulnerability weight, reflecting the cost of falling below that threshold. HM minimizes the weighted, optionally superlinear, awareness shortfall across the population under a fixed seeding budget. Built directly on RASH's continuous, non-monotonic dynamics, HM naturally captures partial protection, recovery, and active suppression, capabilities fundamentally absent from binary, irreversible diffusion models and from classical Influence Maximization (IM) and Influence Minimization (Inf-Min) alike. We prove that the HM optimization problem is NP-hard, while its corresponding harm-reduction function is monotone and $\gamma$-weakly submodular, enabling a greedy algorithm with a bounded approximation guarantee.

In summary, the main contributions of this paper are as follows:

\begin{itemize}
    \item We propose RASH, a signed, susceptibility-aware diffusion model with continuous, bounded, non-monotonic node states that capture both reinforcement and suppression of awareness over time.

    \item We establish that the RASH spread function is monotone and $\gamma$-weakly submodular in networks containing only positive or only negative edges. Thus, although strict submodularity does not generally hold in fully signed networks, greedy optimization retains a provable approximation guarantee in these pure-sign regimes.

    \item We formulate the Harm Minimization (HM) objective, which minimizes vulnerability-weighted awareness shortfall relative to individual safety thresholds, thereby prioritizing protection of individuals at greater risk rather than treating all unreached individuals equally.
    
   \item We prove that HM problem is NP-hard via a reduction from Influence Maximisation under the Independent Cascade (IC) model, which is itself NP-hard through classical reductions from \textsc{Set Cover}. Its harm-reduction function is monotone and $\gamma$-weakly submodular, enabling a greedy algorithm with a provably bounded approximation ratio.

    \item We validate RASH and HM on synthetic and real-world signed temporal networks, demonstrating awareness reversal and suppression beyond classical and signed diffusion baselines, as well as seed allocations that differ substantially from those produced by classical Influence Maximization and Influence Minimization.
\end{itemize}

\section{Related work}
\label{sec:relatedWork}
In this section, we review the evolution of influence optimization research, focusing on classical Influence Maximization, emerging approaches for signed social networks Influence Minimization.

\subsection{Signed Social Networks}

Unlike classical social networks, which are modeled as a graph $G(V,E)$ carrying only positive associations, real-world networks inherently contain both trust and distrust relationships, motivating their representation as a signed graph $G(V,E_p,E_n)$ with disjoint positive and negative edge sets. Two representational conventions have been adopted in the literature: a single signed adjacency matrix with a polarity function mapping each edge to $\{-1,0,+1\}$ \cite{68b}, and an alternative that maintains separate adjacency structures for positive and negative links \cite{133b}. Building on this representation, several structural properties of unsigned networks have been re-examined under signed settings, including clustering coefficient, power-law degree distribution, transitivity, reciprocity, and similarity correlation \cite{PID,135b}, with empirical studies on real-world signed datasets confirming that positive links are denser, more reciprocal, and more transitive than negative ones \cite{136b,31b,35b,152b}. To formalize how such positive and negative ties collectively organize at the group level, two complementary theories have emerged: balance theory, which characterizes socially plausible triadic configurations based on the parity of negative edges \cite{13b}, and status theory, which instead ranks users hierarchically based on the direction and sign of their relationships \cite{69b}. These theoretical foundations have proven useful beyond structural analysis alone, informing tasks such as community detection through polarity-balance maximization \cite{109b}. Collectively, this body of work establishes signed networks as structurally and semantically distinct from their unsigned counterparts, a distinction that becomes especially consequential once we consider how influence itself propagates through such networks.

\subsection{Diffusion Models on Signed Networks}

Building on classical diffusion formulations for unsigned graphs, which propagate only positive influence \cite{4a}, several models have been proposed to capture how both trust and distrust jointly shape information spread. The personal information diffusion model represents each message as carrying fixed absolute properties but relationship-dependent relative properties, with nodes transitioning through unknown, insider, publisher, and forgetter states based on the polarity of the link along which the message arrives \cite{PID}. The signed linear threshold model extends the classical threshold rule by allowing each node to accumulate separate positive and negative influence from its active neighbors, becoming positive- or negative-active only once the net signed influence exceeds a threshold, with polarity fixed once activation occurs \cite{LTE}. Similarly, the signed independent cascade model generalizes the classical single-chance activation rule by making each active node attempt to activate its neighbors with a state matching the polarity of the connecting edge, allowing positive and negative activation to compete across the network \cite{ICE}. In contrast to these state-based models, polarity-related linear influence diffusion avoids Monte Carlo simulation altogether, instead propagating separate positive and negative influence probabilities through linear iterations via damping factors, yielding a more computationally efficient but continuous, rather than discretely-activated, notion of influence \cite{PLID}. Across these models, a common pattern emerges: activation is treated as a discrete, largely irreversible event, and the interaction between positive and negative influence is resolved either by state competition or by a fixed threshold, rather than by modeling how a node's underlying vulnerability or awareness evolves continuously and reversibly over time. This gap in modeling bounded, risk-sensitive, and reversible node behavior under signed diffusion directly motivates the RASH model proposed in this work.

\subsection{Influence Maximization}

Building on these diffusion mechanics, a large body of work has studied how to select seed nodes that maximize propagated influence. Kempe et al. first formalized Influence Maximization (IM) as an NP-hard combinatorial optimization problem and proposed a greedy $(1-\tfrac{1}{e})$-approximation algorithm exploiting submodularity \cite{4a}, with CELF later improving its computational efficiency \cite{11a}. Since approximation guarantees came at high computational cost, subsequent heuristic and sampling-based methods, such as SIMPATH \cite{16a} and IMM \cite{18a}, traded theoretical guarantees for scalability. More recently, learning-based approaches have reframed seed selection using graph representation learning: Khalil et al.'s S2V-DQN combined structure2vec embeddings with deep Q-learning \cite{6a}, while later GNN-based frameworks such as PIANO \cite{32a}, ToupleGDD \cite{33a}, an end-to-end diffusion model DeepIM \cite{34a} and the forecasting based method of seed selection, HoTLink \cite{za2026forecast},  improved both scalability and generalization across graph structures. Beyond the canonical formulation, several variants extend IM toward more targeted or protective objectives, including targeted activation between a source and a specific user \cite{45a} and targeted protection maximization, which instead blocks minimal edges to shield users from harmful influence such as rumors \cite{46a}, foreshadowing the shift from maximizing spread to controlling and containing it.

While the above formulations assume purely cooperative, positive-only relationships, real-world social networks inherently contain both trust and distrust \cite{114b,124b,128b}, under which negative ties can actively suppress or reverse influence, rendering classical IM inapplicable. Ignoring negative relations, as in the original greedy framework, systematically overestimates positive influence spread \cite{4a}, motivating the signed voter-model-based SVIM framework, which extends diffusion to two opposing opinions across balanced, anti-balanced, and unbalanced structures \cite{79b}. Subsequent signed IM research has organized around three strategies trading off guarantees against scalability: simulation-based methods such as the polarity-aware PRIM \cite{76b} and Greedy-SIM \cite{85b}, which preserve submodularity and monotonicity under extended IC-type models at high simulation cost; heuristic methods that rank nodes via centrality or influence scores, such as signed PageRank \cite{149b} and eigenvector-based stubborn-spreader identification \cite{98b}, sacrificing guarantees for efficiency; and sampling-based methods, including reverse-reachable-set techniques like COSiNeMax for maximizing contrasting opinions \cite{102b}, which seek a middle ground between the two. Collectively, these signed IM methods reveal a persistent tension between guarantee-preserving but costly simulation and efficient but heuristic seed selection, motivating models that jointly capture polarity-aware dynamics without sacrificing tractability.

\subsection{Influence Minimization}

The rise of signed networks reshapes not only how influence is maximized but also how it must be contained, since negative edges can actively counteract positive spread and generate interference patterns that unsigned diffusion models cannot capture. This duality introduces new imperatives for influence minimization: protecting vulnerable nodes from negative exposure and disrupting distrust pathways become as important as identifying influential seeds. Most negative influence minimization strategies rely on proactive structural interventions such as vertex blocking, motivated by the practical feasibility of modifying network structure to curb diffusion \cite{25imin,13imin,16imin,20imin}, with candidate vertices identified via degree-based centrality \cite{6imin}, betweenness and out-degree measures \cite{24imin}, or greedy heuristic selection \cite{17imin,22imin}. Beyond vertex-level interventions, alternative containment strategies include edge blocking to sever key propagation channels \cite{5imin} and seeding competing or corrective campaigns, such as countering rumors with verified information \cite{1imin,8imin,15imin}, with more recent formulations further incorporating user experience \cite{16imin} and the temporal evolution of user opinions \cite{12imin} directly into the propagation dynamics.

System-level solutions have also emerged alongside these algorithmic strategies: the Iminimize platform operationalizes vertex-blocking under budget constraints through an interactive interface suited for real-world deployment \cite{teng2023iminimize}, while research on cross-platform misinformation control highlights entity-protection mechanisms that account for users maintaining presences across multiple networks, which otherwise amplifies contagion beyond a single platform \cite{jiang2023deep}. This line of work traces back to Domingos et al., who first modeled influence relationships between users \cite{18ai}, and Kempe et al., who cast such dynamics as the influence maximization problem \cite{4a}; building on the inverse of this formulation, Fan et al. introduced least-cost rumor blocking, seeking a minimal set of protector nodes to curb a rumor's negative impact \cite{5ai}. Subsequent work extended this along several dimensions, incorporating user experience constraints \cite{6ai}, addressing multiple simultaneous rumor cascades via the HISBM model \cite{16ai,17ai}, and solving tree-restricted rumor minimization via dynamic programming \cite{7ai}. 

Across both maximization and minimization, however, existing signed models treat activation as a discrete, largely irreversible event, leaving open the question of how influence should be maximized or contained when node vulnerability and awareness evolve continuously and reversibly, a gap this work directly addresses through RASH and Harm Minimization objective.

\section{Preliminaries}
\label{Sec:Prelims}

This section introduces the basic definitions and main notations  used in our paper. The notations are summarized in Table~\ref{tab:notation}.

\begin{table}[t]
\centering
\caption{Summary of notation used throughout the paper.}
\label{tab:notation}
\footnotesize
\setlength{\tabcolsep}{4pt}
\renewcommand{\arraystretch}{1.1}
\begin{tabularx}{\columnwidth}{@{}l >{\raggedright\arraybackslash}X@{}}
\toprule
\textbf{Symbol} & \textbf{Description} \\
\midrule
$G=(V,E^+,E^-,a)$        & Signed  network with positive/negative edges \\
$a_{ji}$                 & Influence strength from node $j$ to node $i$ \\
$\Gamma_i^+, \Gamma_i^-$ & Positive and negative incoming neighbors of node $i$ \\
$T$                      & Diffusion horizon \\
$S, A$                   & Seed set ($S$: generic/IM context; $A$: HM context) \\
$x_i^S(t)$               & Awareness state of node $i$ at time $t$ under seed set $S$; written $x_i^A(t)$ in the HM context \\
$\sigma(S)$              & Influence spread of seed set $S$ at horizon $T$  \\
$b_i$                    & Personal safety threshold of individual $i$ \\
$w_i$                    & Vulnerability weight of individual $i$ \\
$s_i(A)$                 & Awareness shortfall of individual $i$ under seed set $A$ \\
\bottomrule
\end{tabularx}
\end{table}

\paragraph{Signed Network Model.} We represent a social network as a
directed signed graph $G = (V, E^+, E^-, a)$, where $V$ is the set of
nodes, and $E^+$ and $E^-$ denote the sets of supportive and discouraging
directed edges, respectively. For an edge $(j,i)$, the scalar $a_{ji} > 0$
denotes the strength of influence from node $j$ to node $i$. We write
$\Gamma_i^+$ and $\Gamma_i^-$ for the sets of positive and negative
incoming neighbors of node $i$, i.e., $\Gamma_i^+ = \{j : (j,i) \in E^+\}$
and $\Gamma_i^- = \{j : (j,i) \in E^-\}$. Unlike unsigned graphs, which
admit only $E^+$, this representation allows a single node to
simultaneously receive reinforcing and conflicting signals from different
neighbors, a distinction that is central to everything that follows.

\paragraph{Node States, Seed Sets, and Diffusion Horizon.} Each node
$i \in V$ carries an awareness state that evolves over discrete time steps
$t = 0, 1, \dots, T$, where $T$ is the diffusion horizon. Existing signed 
diffusion models \cite{ICE}\cite{PLID}\cite{PID}\cite{LTE} restrict
this state to a binary or ternary set of discrete values (e.g.,
active/inactive, or positive-active/negative-active/inactive) and treat
activation as irreversible once triggered; we retain this general,
time-indexed notion of a node state without committing to a specific
state space or update rule. A
seed set $S \subseteq V$ denotes the set of nodes initialized with
maximal awareness at $t=0$; we write the resulting, seed-set-dependent
trajectory of node $i$ as $x_i^S(t)$, using the seed set as a superscript
throughout the paper whenever its identity matters. In the Harm
Minimization context introduced next, we use $A \subseteq V$ in the same
role in place of $S$, as noted in Table~\ref{tab:notation}.

\paragraph{Safety and Vulnerability} To motivate
the Harm Minimization (HM) objective, we fix the the notion of safety and vulnerability of individuals, which recur throughout the rest of the paper. Each individual $i \in V$ is associated 
with a personal safety threshold $b_i \in (0,1]$, the minimum awareness
level required to protect them from harm, and a nonnegative vulnerability
weight $w_i \geq 0$, the severity of harm they experience per unit of
awareness shortfall; below $b_i$, individual $i$ is considered
under-informed and at risk, and higher $w_i$ signals that leaving $i$
uninformed is more costly. Both quantities are left abstract here and
instantiated concretely from the structure of the signed graph in
Section~\ref{sec:proposed-work}. Given these, we define a awareness shortfall
that is based on these two parameters :

\begin{definition}[Awareness Shortfall]
Given a seed set $A \subseteq V$, the awareness shortfall of individual
$i$ at the terminal time $T$ is
\[
s_i(A) = \max\{0,\; b_i - x_i^A(T)\},
\]
which is zero once $i$ has received sufficient awareness
($x_i^A(T) \geq b_i$) and strictly positive when they remain
under-protected.
\end{definition}

\section{Proposed Work}
\label{sec:proposed-work}

Having established in Section~\ref{sec:relatedWork} that existing signed
diffusion models treat activation as a discrete, irreversible event, we
now introduce a diffusion mechanism that removes this restriction. We
first define the Risk-Aware Signed Heterogeneous (RASH) diffusion model,
which replaces binary activation with a continuous, reversible awareness
state governed by competing signed influence. We then show that RASH's continuous, reversible dynamics make it
possible to define a harm-sensitive optimization objective directly on
top of it, which we formalize as the Harm Minimization (HM) problem. Finally, we establish the theoretical
properties---monotonicity, weak submodularity, and computational
hardness---that make both influence control and harm minimization under
RASH tractable via greedy approximation.

\subsection{RASH: Risk-Aware Signed Heterogeneous Diffusion Model}
\label{sec:rash}

Real-world diffusion is rarely a matter of pure reinforcement. An
individual embedded in a signed network is simultaneously exposed to
supportive voices that encourage adoption and discouraging voices that
suppress it, and their response to this tension depends on intrinsic
traits such as susceptibility to negative information. RASH is designed
to capture exactly this interplay by modeling awareness as a continuous
quantity that can grow, decay, or reverse under competing influence,
rather than as a one-time, irreversible switch.

We consider a directed signed network $G = (V, E^+, E^-, a)$, in which
each node $i \in V$ is associated with a continuous awareness state
$x_i(t) \in [-1,1]$, where positive values indicate supportive awareness,
negative values indicate discouraging belief, and the magnitude reflects
intensity. Because the dynamics below hold for an arbitrary
initialization, we write this state simply as $x_i(t)$ for now; once a
specific seed set determines that initialization, we reinstate the notation $x_i^S(t)$ fixed in Table~\ref{tab:notation}. We build up the model in five
steps, from the raw signed input a node receives to its updated
awareness. 

\begin{comment}
\begin{figure*}[!htbp]
    \centering
    \includegraphics[width=\linewidth]{Figures/rashprop (1).pdf}
    \caption{Block-level architecture of the RASH diffusion framework. Supportive and discouraging inputs are aggregated using susceptibility-aware weighting, followed by nonlinear normalization and threshold-based activation to update node awareness.}
    \label{fig:rash_block}
\end{figure*}
\end{comment}

\paragraph{Signed influence aggregation} The first step is to aggregate the net influence on a node. A node's positive and negative
neighbors should not be allowed to cancel each other out before the node
has had a chance to weigh them individually: a node with one strong
supporter and one strong detractor is in a meaningfully different
position from a node with no neighbors at all, even though a naively
signed sum would treat both as ``balanced.'' RASH therefore aggregates
the two signals separately. At each time step $t$, node $i$ receives two
separate influence signals from its positive and negative incoming
neighbors $\Gamma_i^+$ and $\Gamma_i^-$:
\begin{equation}
P_i(t) = \sum_{j \in \Gamma_i^+} a_{ji}\, x_j(t),
\label{eq:positive-influence}
\end{equation}
\begin{equation}
N_i(t) = \sum_{j \in \Gamma_i^-} a_{ji}\, x_j(t).
\label{eq:negative-influence}
\end{equation}
$P_i(t)$ aggregates reinforcing signals, while $N_i(t)$ aggregates
opposing or harmful signals; critically, computing these independently
means a node's exposure to support and discouragement never cancels
prematurely at the source.

\paragraph{Susceptibility-based aggregation} The second step is to incorporate the vulnerability factor in updating the awareness of a node. Two nodes receiving
identical $P_i(t)$ and $N_i(t)$ can still respond very differently
depending on how vulnerable each is to discouragement, so a single
combined signal cannot yet be formed without first accounting for this
difference---individuals do not weigh supportive and discouraging signals
equally. We introduce a susceptibility parameter $\lambda_i \in [0,1]$
for each node $i$, where a higher $\lambda_i$ indicates greater
vulnerability to discouraging influence, and combine the two signals as
\begin{equation}
I_i(t) = (1-\lambda_i)\,P_i(t) + \lambda_i\, N_i(t).
\label{eq:aggregated-influence}
\end{equation}
This single parameter lets RASH interpolate continuously between
resilient nodes ($\lambda_i \to 0$, dominated by support) and highly
susceptible nodes ($\lambda_i \to 1$, dominated by discouragement),
rather than forcing every node through identical dynamics.

\paragraph{Nonlinear normalization} The third step is to bound the influence and represent diminishing sensitivity. Left unconstrained, $I_i(t)$ can grow
arbitrarily large simply because a node has many or strongly-weighted
neighbors, which would make the update rule reflect network density
rather than genuine social pressure; a node's sensitivity to additional
influence should also diminish once that influence is already extreme.
To keep influence bounded and represent this diminishing sensitivity, the
aggregated influence is passed through a hyperbolic tangent
transformation:
\begin{equation}
\tilde{I}_i(t) = \tanh\big(I_i(t)\big) \in (-1,1),
\label{eq:normalization}
\end{equation}
which is zero-centered and treats positive and negative influence
symmetrically, so that neither sign is structurally privileged in the
dynamics.

\paragraph{Risk-aware thresholding} The fourth step models how much
pressure is needed before a node actually changes its mind. Not every
fluctuation in influence should perturb a node's belief; if it did,
ordinary noise in $\tilde I_i(t)$ would make awareness trajectories
unrealistically volatile. A node should instead update only once the
pressure it feels is strong enough to matter. Each node $i$ therefore
carries a resistance threshold $\theta_i \in (0,1)$, and its awareness
updates only when the magnitude of normalized influence exceeds this
threshold, $|\tilde I_i(t)| \geq \theta_i$. When this condition holds,
the node is influenced positively if $\tilde I_i(t) > 0$ and negatively
if $\tilde I_i(t) < 0$; otherwise it remains unaffected. The result is that most
nodes, at most time steps, do not update at all, their state changes
only at the specific moments when signed influence is genuinely strong,
rather than at every minor shift in their neighborhood.

\paragraph{Propagation rule} The final step is how the awareness of a node is updated.  A node's belief should not be overwritten
by new influence at every step; it should evolve gradually, retaining
memory of its prior state while remaining open to sufficiently strong new
pressure. The awareness update enforces exactly this balance, combining
memory of the prior state with any new influence that clears the
threshold:
\begin{equation}
x_i(t+1) =
\begin{cases}
(1-\alpha_i)\,x_i(t), & |\tilde I_i(t)| < \theta_i, \\[4pt]
(1-\alpha_i)\,x_i(t) + \alpha_i\, \tilde I_i(t), & \text{otherwise},
\end{cases}
\label{eq:propagation-rule}
\end{equation}
where $\alpha_i \in (0,1)$ is an adaptation rate controlling how quickly
node $i$ updates its state. The term $(1-\alpha_i)x_i(t)$ preserves
persistence, while $\alpha_i \tilde I_i(t)$ injects new influence;
because both reinforcement and degradation act through the same rule,
awareness that has grown under supportive influence can later be eroded
by discouraging influence, a behavior that is structurally impossible
under any signed activation models.  Fig.~\ref{fig:rash_node} illustrates all these steps and mechanism.

\paragraph{Parameterization} Both $\lambda_i$ and $\alpha_i$ should be
estimable directly from the network, rather than requiring hand-tuned
metadata that is rarely available for real signed graphs. A practical
choice for susceptibility $\lambda_i$ is
\begin{equation}
\lambda_i = \frac{|\Gamma_i^-|}{|\Gamma_i^+| + |\Gamma_i^-|},
\label{eq:lambda-estimate}
\end{equation}
capturing a node's structural exposure to discouraging interactions,
i.e., the fraction of $i$'s incoming neighbors that are discouraging.
The adaptation rate $\alpha_i$ is estimated analogously, as a node's
total incoming degree relative to the largest possible incoming degree
in the network i.e.,
\begin{equation}
\alpha_i = \frac{|\Gamma_i^+| + |\Gamma_i^-|}{|V| - 1},
\label{eq:alpha-estimate}
\end{equation}
where $|V|-1$ is the maximum incoming degree attainable by any node in
$G$. This reflects the empirical observation that highly connected
nodes tend to update more conservatively in response to any single
signal: as $|\Gamma_i^+| + |\Gamma_i^-|$ grows relative to $|V|-1$,
$\alpha_i$ approaches $1$, but the influence any single neighbor
contributes to $I_i(t)$ still shrinks in relative terms, since it is
one term among many in the sum defining $P_i(t)$ and $N_i(t)$ 
(Eqs.~\eqref{eq:positive-influence}--\eqref{eq:negative-influence}).

Fig.~\ref{fig:rash_node} walks through this pipeline on a toy signed
network, showing one full update step of the RASH Diffusion Engine from
time stamp $t$ to time stamp $t{+}1$. The input is a risk/sign matrix
recording, for each ordered pair of users, whether their relationship is
supportive ($+$), discouraging ($-$), or absent, together with the
associated edge weight $a_{ji}$; this matrix is what determines each
node's positive and negative incoming neighborhoods $\Gamma_i^+$ and
$\Gamma_i^-$. Starting from this signed structure and each node's
awareness at time $t$, the figure traces the five steps of the pipeline
in order and update the state of unaware node at $t{+}1$. Nodes are color-coded throughout as
unaware, positively aware, or negatively aware to  show how node's classification can shift between the two time
stamps once its updated awareness crosses zero.

\begin{figure*}[!htbp]
    \centering
    \includegraphics[width=\linewidth]{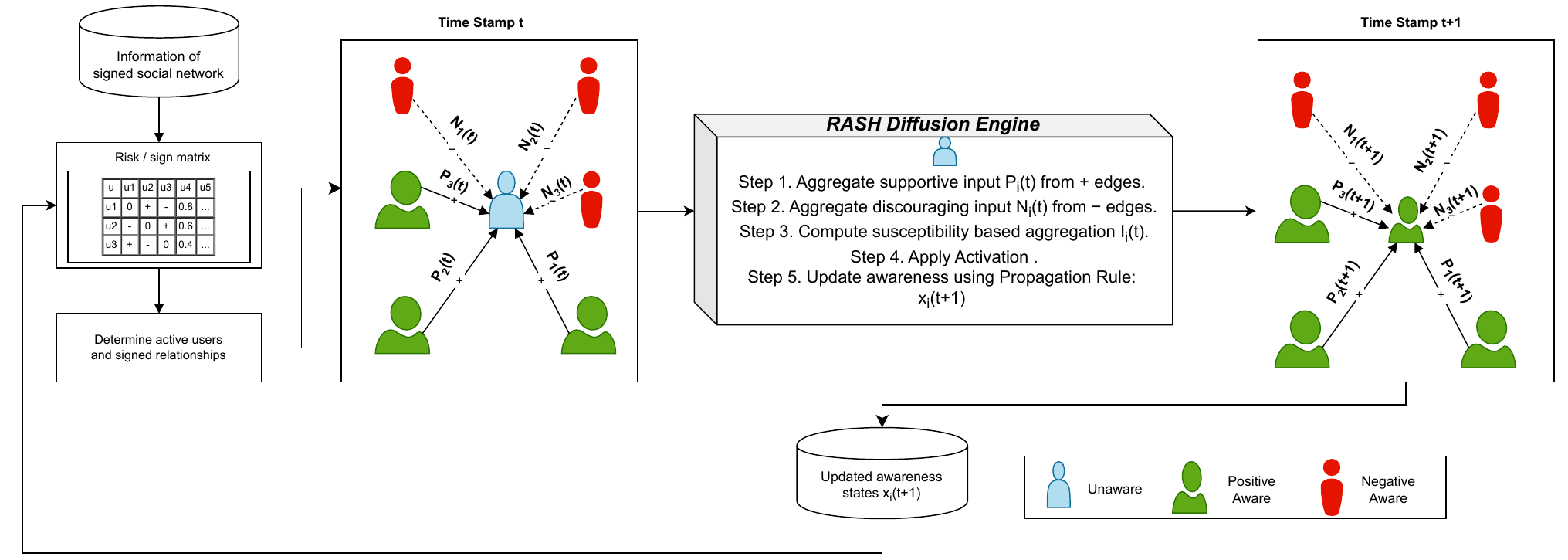}
    \caption{Node-level illustration of the RASH propagation mechanism. A node receives supportive (solid) and discouraging (dashed) influence from neighbors, which are aggregated, normalized, and used to update awareness. \footnotesize\textit{Note: the adjacency matrix shown in the figure is used solely for illustrative purposes to represent the sign and risk information, and does not correspond to the exact adjacency matrix of any real-world social network.}}
    \label{fig:rash_node}
\end{figure*}

\subsection{Harm Minimization (HM)}
\label{sec:hm}

The conventional influence objectives optimize aggregate propagation, but they do not distinguish between individuals according to their susceptibility to harm or the level of awareness required for their protection. RASH enables a different perspective by coupling continuous, signed diffusion with an objective that explicitly accounts for individual-level risk.

\subsubsection{From Influence to Harm Minimization}

Because the diffusion dynamics and the seed-selection objective are
separable in RASH (Section~\ref{sec:rash}), the classical influence
objectives can be recovered as special cases. For a seed set $S\subseteq
V$ and diffusion horizon $T$, let
\begin{equation}
\sigma(S)=\sum_{i\in V}x_i^S(T)
\end{equation}
denote the aggregate terminal awareness under the RASH dynamics. In the
positive regime (with only positive edges), obtained by setting
$\lambda_i=0$ for all $i$, the resulting spread function $\sigma^+$
gives the conventional Influence Maximization (IM) problem,
\begin{equation}
\max_{S\subseteq V,\;|S|\leq k}\sigma^+(S).
\end{equation}

Conversely, the negative regime (with only negative edges), obtained by
setting $\lambda_i=1$ for all $i$, yields the influence-reduction
setting. Here the goal is not to seed awareness but to place a blocking
set that suppresses it, so the natural objective is not the raw spread
$\sigma^-(S)$ itself but how much that spread is \emph{reduced} relative
to seeding nothing at all: a blocking set is good precisely when
negative spread under it is much smaller than the negative spread that
would occur unblocked. If $\sigma^-(S)$ denotes the resulting negative
influence under blocking set $S$, this reduction relative to the
unblocked baseline $\sigma^-(\emptyset)$ is
\begin{equation}
\varrho(S)=\sigma^-(\emptyset)-\sigma^-(S),
\end{equation}
and Influence Minimization selects the blocking set of size at most $k$
maximizing this reduction,
\begin{equation}
\max_{S\subseteq V,\;|S|\leq k}\varrho(S).
\end{equation}

Both objectives, however, evaluate a campaign primarily through
aggregate influence. Consequently, they implicitly treat the value of
reaching different individuals as homogeneous. This assumption is
inadequate for protective campaigns such as public-health warnings,
cybersecurity advisories, and crisis communications, where failing to
protect a highly vulnerable individual can be substantially more
consequential than failing to reach an otherwise low-risk individual.

We therefore formulate Harm Minimization (HM) around the residual
protection deficit of each individual rather than aggregate reach. Let
$b_i$ denote the minimum awareness required for node $i$ to be
considered adequately protected (safety threshold), and let $w_i$
quantify the vulnerability of that node. Given a seed set $A$, the
resulting awareness shortfall at the end of the diffusion horizon is
\begin{equation}
s_i(A)=\max\{0,b_i-x_i^A(T)\}.
\label{eq:awareness-shortfall}
\end{equation}
This shortfall is zero whenever node $i$'s terminal awareness
$x_i^A(T)$ reaches or exceeds its threshold $b_i$, meaning $i$ has
received enough support to be considered protected regardless of how
far above $b_i$ it ends up; it is strictly positive, and grows linearly
with the gap $b_i - x_i^A(T)$, whenever $i$ falls short, meaning partial
progress toward $b_i$ is credited but only full attainment removes the
deficit entirely. Unlike $b_i$ and $w_i$, which are node-level
characteristics fixed before seed selection, $s_i(A)$ depends directly
on the selected seed set, since $x_i^A(T)$ is itself the outcome of
running RASH from $A$: two different seed sets can leave the same node
with two different shortfalls, even though its threshold $b_i$ never
changes.

The Harm Minimization problem then asks for the seed set that leaves
the population with the least total, vulnerability-weighted shortfall.
Formally, given a budget $k$, HM selects
\begin{equation}
\min_{A \subseteq V,\, |A| \le k} \; \sum_{i \in V} w_i \, \big(s_i(A)\big)^p,
\label{eq:hm-problem}
\end{equation}
for exponent $p \ge 1$. Each term in this sum penalizes exactly one
individual's remaining deficit, scaled by how costly that individual's
neglect is; summing over $i \in V$ means the objective is sensitive not
just to how many individuals remain under-protected, but to how severely
and how expensively each one does. A seed set that reaches many
low-vulnerability nodes while leaving a few high-vulnerability nodes far
below their threshold can therefore score worse under
Eq.~\eqref{eq:hm-problem} than a seed set that reaches fewer nodes
overall but closes the gap for those who need it most---precisely the
behavior aggregate-reach objectives like IM and Inf-Min cannot express,
since neither $\sigma^+(S)$ nor $\varrho(S)$ distinguishes which
individuals were reached, only how many. HM therefore evaluates a
campaign according to how effectively it reduces protection deficits
across the population, not according to how much awareness it produces
in total.

\subsubsection{Risk-Aware Protection Requirements}

The vulnerability weight should be consistent with the mechanism
through which RASH generates discouraging influence. Since $\lambda_i$
already quantifies node $i$'s susceptibility to negative influence, we
use it to model the vulnerability weight and keep $w_i = \lambda_i$ for
the sake of simplicity. This coupling ensures that individuals who are
more susceptible to discouraging influence receive greater importance
in the harm objective. The protection requirement $b_i$ captures the level of awareness node
$i$ must reach before being considered safe, and this requirement
should not be uniform across the population: a node already surrounded
by discouraging influence, or one with few alternative channels through
which to receive or recover awareness, needs more of it before it can
be regarded as protected. We derive $b_i$ from exactly these two
complementary structural risk factors, exposure to negative influence
and structural isolation, each computed directly from the positive and
negative neighborhoods $\Gamma_i^+$ and $\Gamma_i^-$, so that no
metadata beyond the signed graph itself is required. The first factor
is already at hand: it is exactly the susceptibility term $\lambda_i$
used above for $w_i$, since $\lambda_i$ is itself the fraction of $i$'s
incoming neighbors that are discouraging. The second factor requires a
new quantity, the \emph{structural isolation} $\iota_i$, measuring $i$'s
incoming degree relative to the largest incoming degree in the network,
subtracted from $1$ so that sparsely connected nodes receive a higher
value:
\[
\iota_i=
1-
\frac{|\Gamma_i^+|+|\Gamma_i^-|}
{\displaystyle\max_{j\in V}
\left(|\Gamma_j^+|+|\Gamma_j^-|\right)}.
\]
The two quantities capture distinct and independent notions of risk:
$\lambda_i$ reflects how much of a node's \emph{existing} signal is
already working against it, while $\iota_i$ reflects how few
\emph{alternative} structural channels remain through which awareness
could still be acquired or recovered if the primary campaign path fails
to reach it.A node can score highly on the negative-exposure dimension
($\lambda_i$) without scoring highly on the structural-isolation
dimension ($\iota_i$), or vice versa, so neither factor alone is
sufficient to characterize its risk. Reusing $\lambda_i$ here, rather than introducing
a separate exposure term, is deliberate: it means $b_i$ and $w_i$ are
not two independently chosen quantities but share a common structural
ingredient by construction, so a node's susceptibility shapes both how
urgently it must be protected and how costly it is to fail.

We combine the two through a convex combination and map the resulting
risk score onto a bounded protection threshold:
\begin{equation}
b_i=
b_{\min}+
(b_{\max}-b_{\min})
\left[
\beta\,\lambda_i+(1-\beta)\,\iota_i
\right],
\qquad
\beta\in[0,1],
\label{eq:b_i}
\end{equation}
where $0 < b_{\min} \leq b_{\max} \leq 1$ are fixed constants ensuring
every node's required awareness level stays within a nonzero lower and
bounded upper limit, and $\beta$ controls the relative weight given to
negative exposure versus structural isolation. Under this construction,
nodes that are more exposed to negative influence, more structurally
isolated, or both, are assigned a higher $b_i$, and therefore require a
higher level of terminal awareness before being regarded as adequately
protected.
\subsubsection{Need-Weighted Harm Objective}

Given the node-specific vulnerability $w_i$ and protection threshold $b_i$, we define the total harm associated with a seed set $A$ by aggregating individual awareness shortfalls:

\begin{equation}
H_p(A)
=
\sum_{i\in V}
w_i\left(s_i(A)\right)^p,
\qquad p\geq1.
\label{eq:need-weighted-harm}
\end{equation}

The exponent $p$ controls how strongly the objective penalizes concentrated shortfalls. When $p=1$, harm grows linearly with the protection deficit. Increasing $p$ places progressively greater emphasis on individuals experiencing severe shortfalls, thereby discouraging solutions that achieve good aggregate protection while leaving a small number of individuals substantially under-protected.

The resulting Harm Minimization problem is

\begin{equation}
A^*
=
\arg\min_{A\subseteq V,\;|A|\leq k}
H_p(A).
\label{eq:hm-objective}
\end{equation}

Thus, rather than maximizing the total amount of information propagated or minimizing aggregate negative influence, HM selects seeds to minimize the weighted residual protection deficit of the population. The seed budget $k$ determines the available intervention capacity, while $w_i$, $b_i$, and $p$ determine how that capacity is distributed across individuals with different levels of vulnerability and protection requirements.

This formulation also provides a direct policy interpretation of $p$. The case $p=1$ corresponds to proportional penalization of individual shortfalls, whereas $p>1$ increasingly prioritizes the most severely under-protected individuals. In the limiting regime $p\rightarrow\infty$, the objective increasingly approaches minimization of the largest awareness shortfall $
s_i(A)$, providing a strong worst-case protection criterion.

\begin{table}[t]
\centering
\caption{Comparison of seed-selection objectives under RASH.}
\label{tab:objective-comparison}
\footnotesize
\setlength{\tabcolsep}{3pt}
\renewcommand{\arraystretch}{1.2}
\begin{tabularx}{\columnwidth}{@{}l >{\raggedright\arraybackslash}X >{\raggedright\arraybackslash}X@{}}
\toprule
\textbf{Objective} & \textbf{Optimization goal} & \textbf{Primary criterion} \\
\midrule
Influence Maximization
& Maximize positive influence
& Aggregate awareness/reach \\

Influence Minimization
& Maximize reduction in negative influence
& Aggregate negative influence \\

Harm Minimization
& Minimize $\sum_{i} w_i(s_i(A))^p$
& Individual protection deficits \\
\bottomrule
\end{tabularx}
\end{table}

\subsubsection{Role of RASH in Harm Minimization}

The HM objective is intrinsically tied to the properties of the RASH diffusion process. First, RASH produces continuous awareness states, allowing the shortfall $s_i(A)$ to distinguish between partial and severe under-protection. This information would be lost in a purely binary activation model. Second, RASH permits reversible and non-monotonic awareness dynamics, allowing a node's awareness to decrease as a consequence of discouraging influence. Consequently, HM can capture not only the failure to reach a vulnerable individual but also the possibility that an intervention leaves that individual worse off. Third, the definition $w_i=\lambda_i$ directly couples individual vulnerability in the objective with susceptibility to discouraging influence in the diffusion process.

These properties establish RASH as the diffusion model underlying HM:
RASH determines how awareness evolves, while the HM objective
determines how seed selection should account for heterogeneous
protection requirements and vulnerability. The resulting formulation
therefore shifts the goal of influence-based intervention from not only
maximising aggregate reach but also minimising the residual harm
experienced across the population.

\subsection{Properties of the Harm Minimization Objective under RASH}
\label{sec:hm-props}

Having defined RASH and the HM objective, we now establish that both
admit tractable approximation guarantees, despite the added complexity
introduced by signed, non-monotonic dynamics. We establish these
guarantees in the positive regime of RASH ($\lambda_i = 0$ for all
$i$, so that $I_i(t) = P_i(t)$), together with $a_{ji} \ge 0$ for all
$i,j$ and $x_i^S(0) \in [0,1]$ for all $i, S$. By symmetry, replacing
every occurrence of $P_i$ with $N_i$ and $\Gamma_i^+$ with
$\Gamma_i^-$ throughout the definitions and proofs below, an entirely
analogous argument yields matching guarantees in the negative
regime.

The fully general signed regime, in which $\lambda_i$ varies freely
across $(0,1)$, is precisely where strict submodularity fails,
and it is worth seeing concretely why, rather than taking this on
faith, since it motivates why a weaker guarantee is what we
prove below.

\begin{example}[Failure of Strict Submodularity]
\label{ex:submod-failure}
Let node $i$ have two positive incoming neighbors $v_1, v_2$, with
$a_{v_1 i} = a_{v_2 i} = 0.4$, resistance threshold $\theta_i = 0.5$,
adaptation rate $\alpha_i = 1$, and $x_i(0) = 0$. Take $A = \emptyset$
and $B = \{v_2\}$, so $A \subseteq B$, and consider adding $v_1$ to
each.

Adding $v_1$ to $A$: with only $v_1$ seeded, $P_i(0) = a_{v_1 i}\cdot 1
= 0.4$, so $\tilde I_i(0) = \tanh(0.4) \approx 0.38 < \theta_i$. The
threshold is not cleared, so $x_i(1) = x_i(0) = 0$: the marginal gain
of adding $v_1$ to $A$ is exactly $0$.

Adding $v_1$ to $B$: with both $v_1$ and $v_2$ seeded, $P_i(0) =
0.4+0.4 = 0.8$, so $\tilde I_i(0) = \tanh(0.8) \approx 0.66 \ge
\theta_i$. The threshold is now cleared, and with $\alpha_i=1$,
$x_i(1) = \tilde I_i(0) \approx 0.66$: the marginal gain of adding
$v_1$ to $B$ is approximately $0.66$.

The marginal gain of adding $v_1$ therefore \emph{increased} when the
seed set grew from $A=\emptyset$ to $B=\{v_2\} \supseteq A$, the exact
reverse of the diminishing-returns behavior strict submodularity
requires. This complementarity arises directly from the threshold
gate: whether $v_1$'s contribution matters at all depends on whether
other already-selected seeds have pushed $i$ close enough to
$\theta_i$ for $v_1$ to tip it over. In the fully signed regime, the
same mechanism interacts with sign: since $P_i(t)$ and $N_i(t)$
combine into a single gated quantity $I_i(t)$ before the threshold
test (Eq.~\eqref{eq:aggregated-influence}), a positive seed's marginal
effect on $i$ depends on how much discouraging pressure other selected
seeds are simultaneously exerting on $i$, a dependency with no reason
to be monotonically diminishing, and one that strict submodularity
cannot accommodate.
\end{example}

We therefore prove a relaxed, multiplicative version of diminishing
returns instead ($\gamma$-weak submodularity, Section~\ref{sec:hm-props}
below), which still suffices for the greedy algorithm to carry a provable approximation
guarantee. We build the argument in three stages: a coupling lemma
bounding how a node's response to a small seed set relates to its
response to a larger one, monotonicity of harm reduction, and
$\gamma$-weak submodularity itself.

\subsubsection{Setup and Notation}

Since minimizing harm and maximizing harm \emph{reduction} are
equivalent problems over the same seed sets, it is more convenient to
argue in terms of the latter, which behaves like a monotone gain
function in the sense familiar from submodular optimization. We
therefore work throughout with:
\begin{itemize}
  \item $R_p(A) := H_p(\emptyset) - H_p(A)$: the harm reduction achieved
  by seed set $A\subseteq V$, relative to seeding nothing.
  \item $H_p(A) := \sum_{i\in V} w_i\,(s_i(A))^p$: total harm under $A$,
  as defined in Eq.~\eqref{eq:need-weighted-harm}; $w_i\ge0$ is node
  $i$'s vulnerability weight and $p\ge1$ the harm exponent.
  \item $s_i(A) := \max\{0,\,b_i - x_i^A(T)\}$: the shortfall of node
  $i$ at horizon $T$ relative to threshold $b_i$.
\end{itemize}
Because $H_p(\emptyset)$ does not depend on $A$, minimizing $H_p(A)$
over $|A|\le k$ is equivalent to maximizing $R_p(A)$ over $|A|\le k$;
all guarantees are stated for $R_p$ and translate directly to
$H_p$.

To bound how much a node's
state can change, $|x_i^{S}(t{+}1)|$, in terms of how much its
input changed, $|I_i^S(t)|$. Under the piecewise propagation
rule of Eq.~\eqref{eq:propagation-rule}, this is not possible in
general: the rule is discontinuous exactly at $|\tilde I_i(t)| =
\theta_i$, so an arbitrarily small change in input can push a node
from "unaffected" to "fully updated," producing an output change with
no bound in terms of the input change at all. To rule this out, we
replace the hard gate with a smoothed one and require it to be
continuously differentiable. This lets us to invoke the
Mean Value Theorem: for a $C^1$ function $F_i$ and any $0 \le a
\le b$, the theorem guarantees $F_i(b) - F_i(a) = F_i'(\xi)\,(b-a)$ for
some $\xi \in (a,b)$, so a bound on $F_i'$ over the relevant range
translates immediately into a bound on the difference $F_i(b)-F_i(a)$
itself, exactly the tool needed to control $|x_i^S(t{+}1)|$ from a
bound on $|I_i^S(t)|$.

\begin{assumption}[Lipschitz Threshold Gate]
\label{ass:gate}
The activation update is
\begin{equation}
  x_i^S(t{+}1) = (1-\alpha_i)\,x_i^S(t)
  + \alpha_i\,F_i\big(I_i^S(t)\big),
  \label{eq:gate}
\end{equation}
where
\begin{equation}
  I_i^S(t) := \sum_{j\in\Gamma_i^+} a_{ji}\,x_j^S(t)
\end{equation}
is the total input to node $i$---which, since $\lambda_i = 0$,
coincides with $P_i^S(t)$ of Eq.~\eqref{eq:positive-influence}---and
\begin{equation}
  F_i(I) := \phi\big(\kappa_i(\tanh(I)-\theta_i)\big)\tanh(I)
\end{equation}
is a smoothed gate, with $\phi(z) = 1/(1+e^{-z})$ the logistic function,
$\theta_i \in (0,1)$ the resistance threshold of
Eq.~\eqref{eq:propagation-rule}, and $\kappa_i \in (0,\infty)$ a
steepness parameter controlling how sharply the gate transitions. As
$\kappa_i\to\infty$, $F_i$ recovers the original hard-threshold rule
exactly. For every finite $\kappa_i>0$, $F_i:[0,\infty)\to[0,1)$ is
$C^1$, non-decreasing, and $F_i(0)=0$.
\end{assumption}

\begin{definition}[Constants]
\label{def:constants}
The following quantities are used throughout the remainder of this
subsection. $M$ and $R^+$ bound the raw input a node can receive;
$L_i^-, L_i^+$ bound the gate's sensitivity to that input; $\gamma^+$
combines these into the single coupling constant used in the coupling
lemma and both theorems below.

\begin{center}
\renewcommand{\arraystretch}{1.3}
\begin{tabular}{@{}l l p{4.6cm}@{}}
\toprule
\textbf{Symbol} & \textbf{Definition} & \textbf{Role} \\
\midrule
$\alpha_{\min}$ & $\min_{i\in V}\alpha_i$ & slowest update rate across all nodes \\
$M$ & $\max_{i,j} a_{ji}$ & largest influence strength of any single edge \\
$R^+$ & $M\,|V|$ & bound on any input $I_i^S(t)$ that can arise\footnotemark \\
$L_i^-$ & $\min_{I\in[0,R^+]} F_i'(I)$ & node $i$'s minimum gate sensitivity on $[0,R^+]$ \\
$L_i^+$ & $\max_{I\in[0,R^+]} F_i'(I)$ & node $i$'s maximum gate sensitivity on $[0,R^+]$ \\
$L^-$ & $\min_{i\in V} L_i^-$ & network-wide minimum gate sensitivity \\
$\gamma^+$ & $\alpha_{\min}\,L^-$ & coupling constant used below \\
\bottomrule
\end{tabular}
\end{center}
\footnotetext{Since $I_i^S(t) = \sum_{j\in\Gamma_i^+} a_{ji}x_j^S(t)
\le M|\Gamma_i^+|$ by Lemma~\ref{lem:prelim}(a), and $|\Gamma_i^+| \le
|V|-1$, the tighter bound $M(|V|-1)$ also holds; we use the slightly
looser $M|V|$ for notational convenience.}
\end{definition}

Positivity of $L^-$---and hence of $\gamma^+$---is not assumed above
but proved next, as a consequence of Assumption~\ref{ass:gate}.

\begin{remark}[Positivity of the Gate Derivative]
\label{rem:gate-positivity}
Write $F_i(I) = \phi(\psi(I))\tanh(I)$, where
\begin{equation}
  \psi(I) := \kappa_i\bigl(\tanh(I) - \theta_i\bigr)
  \label{eq:psi}
\end{equation}
measures how far the normalized influence $\tanh(I)$ lies from the
resistance threshold $\theta_i$, scaled by the steepness parameter
$\kappa_i$; $\psi(I)$ is the argument the logistic gate $\phi$ acts on.
Differentiating $F_i$ by the chain and product rules gives
\begin{equation}
  F_i'(I) = \operatorname{sech}^2(I)
  \Big[\underbrace{\kappa_i\,\phi'(\psi(I))\,\tanh(I)}_{\text{(i)}}
  + \underbrace{\phi(\psi(I))}_{\text{(ii)}}\Big].
  \label{eq:gate-derivative}
\end{equation}
On $I \ge 0$: $\operatorname{sech}^2(I) > 0$ everywhere; term (i) is
nonnegative, since $\kappa_i>0$, $\tanh(I)\ge0$, and $\phi'(z) =
\phi(z)(1-\phi(z)) > 0$ for all $z$; term (ii) is strictly positive,
since $\phi(\psi(I)) \in (0,1)$ for any finite argument. Hence the
sum (i)+(ii), and therefore $F_i'(I)$, is strictly positive for every
$I \in [0,R^+]$. Since $F_i \in C^1$,
$F_i'$ is continuous on the compact interval $[0,R^+]$, so it attains
a minimum there, and that minimum, $L_i^-$, is strictly positive;
consequently $L^- = \min_i L_i^- > 0$ and $\gamma^+ = \alpha_{\min}L^->
0$.
\end{remark}
\subsubsection{Boundedness and Marginal Coupling}

Before comparing how seed sets of different sizes affect a node, we
first need two more basic facts. The first fact that awareness never leaves the valid
range $[0,1]$ regardless of seeding, and that the effect of adding one
extra seed to a small set is never smaller, up to a controlled factor,
than its effect when added to a larger superset. The second fact that a single seed's
marginal contribution decays in a structured, boundable way as the
seed set grows, which is exactly the property needed to later establish
weak submodularity.

\begin{lemma}[Boundedness and Marginal Coupling]
\label{lem:prelim}
Let $G=(V,E^+,E^-,a)$ be a signed network with $a_{ji}\ge0$ for all
$i,j$, and suppose every node $i \in V$ updates according to
\begin{equation}
  x_i^S(t{+}1) = (1-\alpha_i)\,x_i^S(t) + \alpha_i\,F_i\big(I_i^S(t)\big),
\end{equation}
where $I_i^S(t) := \sum_{j\in\Gamma_i^+} a_{ji}\,x_j^S(t)$ and
$F_i:[0,\infty)\to[0,1)$ is $C^1$, non-decreasing, and satisfies
$F_i(0)=0$, with $x_i^S(0)\in[0,1]$ for every seed set $S$ and every
$i \in V$. Then the following two properties hold.

\smallskip
\noindent\textup{(a) Boundedness.} For every seed set $S \subseteq V$,
every node $i \in V$, and every time step $t \ge 0$,
\begin{equation}
  x_i^S(t) \in [0,1].
\end{equation}

\smallskip
\noindent\textup{(b) Marginal Coupling.} For every pair of seed sets
$A \subseteq B \subseteq V$ and every node $v \in V \setminus B$, define
the marginal effect of adding $v$ to $A$ and to $B$ respectively, for
every node $i \in V$ and every time step $t \ge 0$, as
\[
\delta_i^A(t) := x_i^{A\cup\{v\}}(t) - x_i^A(t),
\qquad
\delta_i^B(t) := x_i^{B\cup\{v\}}(t) - x_i^B(t).
\]

Then, for every such $A, B, v, i,$ and $t$,

\[
\delta_i^A(t) \;\geq\; (\gamma^+)^t\, \delta_i^B(t).
\]
\end{lemma}
\begin{proof}
We prove (a), then two intermediate claims, then (b).

\smallskip
\noindent\textit{Part (a): Boundedness.} Induction on $t$.

Base case $t=0$: $x_i^S(0)\in[0,1]$ by hypothesis.

Inductive step: assume $x_j^S(t)\in[0,1]$ for all $j$. Then
\begin{equation}
  I_i^S(t) = \sum_{j\in\Gamma_i^+} a_{ji}\,x_j^S(t) \;\ge\; 0,
\end{equation}
since $a_{ji}\ge0$, $x_j^S(t)\ge0$. Since $\phi(\cdot)\in(0,1)$ and
$\tanh(I)\in[0,1)$ for $I\ge0$,
\begin{equation}
  F_i\big(I_i^S(t)\big) \in [0,1).
\end{equation}
By \eqref{eq:gate}, $x_i^S(t{+}1)$ is a convex combination of
$x_i^S(t)\in[0,1]$ and $F_i(I_i^S(t))\in[0,1)$, so
\begin{equation}
  x_i^S(t{+}1) \in [0,1],
\end{equation}
closing the induction.

\smallskip
\noindent\textit{Claim 1 (State Monotonicity):}
$A\subseteq B \Rightarrow x_i^A(t)\le x_i^B(t)$ for all $i,t$.

Base case: holds at $t=0$ by initialization.

Inductive step: assume $x_j^A(t)\le x_j^B(t)$ for all $j$. Since
$a_{ji}\ge0$,
\begin{equation}
  I_i^A(t) \;\le\; I_i^B(t).
\end{equation}
Since $F_i$ is non-decreasing,
\begin{equation}
  F_i\big(I_i^A(t)\big) \;\le\; F_i\big(I_i^B(t)\big).
\end{equation}
Applying \eqref{eq:gate} to both trajectories,
\begin{equation}
  x_i^A(t{+}1) \;\le\; x_i^B(t{+}1),
\end{equation}
closing the induction. Consequently,
\begin{equation}
  \delta_i^A(t) \ge 0, \qquad \delta_i^B(t) \ge 0
  \qquad \forall i,t.
\end{equation}

\smallskip
\noindent\textit{Claim 2 (Local Lipschitz Bound):} For
$0\le a\le b\le R^+$,
\begin{equation}
  L^-(b-a) \;\le\; F_i(b) - F_i(a) \;\le\; L_i^+(b-a).
  \label{eq:mvt}
\end{equation}
By part (a) and Definition~\ref{def:constants}, every trajectory in (b)
has $I_i^{(\cdot)}(t)\in[0,R^+]$. Since $F_i\in C^1$ with
$F_i'\in[L_i^-,L_i^+]$ on $[0,R^+]$ (Remark~\ref{rem:gate-positivity}),
the Mean Value Theorem gives \eqref{eq:mvt}, using $L^-\le L_i^-$.

\smallskip
\noindent\textit{Part (b).} Induction on $t$.

Base case $t=0$: node $v$ has just been added, so
\begin{equation}
  \delta_i^A(0) = \mathbf{1}[i=v] = \delta_i^B(0),
\end{equation}
and marginal coupling holds with equality.

Inductive step: assume Marginal coupling holds at time $t$ for
every $i$. Define
\begin{equation}
  \Delta I_i^A(t) := \sum_{j\in\Gamma_i^+} a_{ji}\,\delta_j^A(t),
\end{equation}
\begin{equation}
  \Delta I_i^B(t) := \sum_{j\in\Gamma_i^+} a_{ji}\,\delta_j^B(t).
\end{equation}
By the inductive hypothesis and $a_{ji}\ge0$,
\begin{equation}
  \Delta I_i^A(t) \;\ge\; (\gamma^+)^t\, \Delta I_i^B(t).
  \label{eq:DeltaI}
\end{equation}
Applying \eqref{eq:gate} to $(B,B\cup\{v\})$ with the upper bound
in \eqref{eq:mvt},
\begin{equation}
  \delta_i^B(t{+}1) \;\le\;
  (1-\alpha_i)\,\delta_i^B(t) + \alpha_i\,\Delta I_i^B(t).
  \label{eq:ub}
\end{equation}
Applying \eqref{eq:gate} to $(A,A\cup\{v\})$ with the lower bound
in \eqref{eq:mvt},
\begin{equation}
  \delta_i^A(t{+}1) \;\ge\;
  (1-\alpha_i)\,\delta_i^A(t) + \alpha_i L^-\,\Delta I_i^A(t).
  \label{eq:step1}
\end{equation}
Substituting the inductive hypothesis and \eqref{eq:DeltaI} into
\eqref{eq:step1},
\begin{equation}
  \delta_i^A(t{+}1) \;\ge\;
  (\gamma^+)^t\Big[(1-\alpha_i)\delta_i^B(t)
  + \alpha_i L^- \Delta I_i^B(t)\Big].
  \label{eq:step2}
\end{equation}
Since $\alpha_i\ge\alpha_{\min}$, $\alpha_i L^-\ge\gamma^+$, so
\eqref{eq:step2} gives
\begin{equation}
  \delta_i^A(t{+}1) \;\ge\;
  (\gamma^+)^t\Big[(1-\alpha_i)\delta_i^B(t)
  + \gamma^+ \Delta I_i^B(t)\Big].
  \label{eq:step3}
\end{equation}
Set $p:=(1-\alpha_i)\delta_i^B(t)\ge0$ and
$q:=\Delta I_i^B(t)\ge0$. Comparing \eqref{eq:ub} and
\eqref{eq:step3} requires
\begin{equation}
  p + \gamma^+ q \;\ge\; \gamma^+\big(p + \alpha_i q\big).
\end{equation}
Rearranging,
\begin{equation}
  (p+\gamma^+q) - \gamma^+(p+\alpha_iq)
  = p(1-\gamma^+) + \gamma^+(1-\alpha_i)q.
  \label{eq:final-ineq}
\end{equation}
Since $\gamma^+\le\alpha_{\min}\le\alpha_i\le1$, both
$(1-\gamma^+)$ and $(1-\alpha_i)$ are $\ge0$; with $p,q\ge0$, the
right side of \eqref{eq:final-ineq} is $\ge0$. Hence
\begin{equation}
  \delta_i^A(t{+}1) \;\ge\; (\gamma^+)^{t+1}\,\delta_i^B(t{+}1),
\end{equation}
closing the induction and proving the marginal coupling.
\end{proof}

\subsubsection{Monotonicity.} Adding seeds should never make the population
worse off: since a node's awareness can only move upward, never
downward, as more seeds are added (Lemma~\ref{lem:prelim}), the harm
$H_p(A)$ can only decrease, never increase, as the seed set $A$ grows.
This property does not follow from the algebraic form of $H_p$ alone;
it is a consequence of how RASH propagates awareness, and must be
established as such.

\begin{proposition}[Monotonicity of Harm Reduction]
\label{prop:mono-hm}
Given the network and RASH propagation rule for all $A\subseteq B\subseteq V$,
\begin{equation}
  R_p(A) \le R_p(B), \qquad\text{equivalently}\qquad
  H_p(A) \ge H_p(B).
\end{equation}
\end{proposition}

\begin{proof}
 By Lemma~\ref{lem:prelim}(b), applied with $B\setminus
A$ added one node at a time, $x_i^A(T) \le x_i^B(T)$ for every $i\in V$.

Since $s_i(A)$ is non-increasing in $x_i^A(T)$,
\begin{equation}
  s_i(A) \ge s_i(B) \qquad \forall i\in V.
\end{equation}

 Since $z\mapsto z^p$ is non-decreasing for
$z\ge0$, $p\ge1$,
\begin{equation}
  \big(s_i(A)\big)^p \ge \big(s_i(B)\big)^p \qquad \forall i\in V.
\end{equation}

 Since $w_i\ge0$, summing over $i$,
\begin{equation}
  H_p(A) = \sum_{i\in V} w_i\big(s_i(A)\big)^p
  \;\ge\; \sum_{i\in V} w_i\big(s_i(B)\big)^p = H_p(B).
\end{equation}

 Subtracting from $H_p(\emptyset)$,
\begin{equation}
  R_p(A) \;\le\; R_p(B).
\end{equation}
\end{proof}

\subsubsection{Weak Submodularity}Monotonicity alone does not provide a greedy approximation guarantee.
Although strict submodularity does not hold in the general signed
setting, the reward $R_p$ satisfies a multiplicative diminishing-returns
property known as weak submodularity.

\begin{definition}[$\gamma$-Weak Submodularity]
\label{def:weaksub}
A monotone function $f:2^V\to\mathbb{R}_{\ge0}$ is $\gamma$-weakly
submodular, for $\gamma\in(0,1]$, if for all $A\subseteq B\subseteq V$
and all $U\subseteq V\setminus B$,
\begin{equation}
  \sum_{v\in U} \big[f(A\cup\{v\}) - f(A)\big]
  \;\ge\; \gamma\cdot\big[f(B\cup U) - f(B)\big].
  \label{eq:weaksub-def}
\end{equation}
\end{definition}

\begin{theorem}[Weak Submodularity of Harm Reduction]
\label{thm:weaksub-hm}
Given the signed network and RASH propagation rule 
with $R_p$ as in Proposition~\ref{prop:mono-hm}, $R_p$ is
$\gamma$-weakly submodular
with
\[
  \gamma = (\gamma^+)^T \;>\; 0,
\]
$T$ the diffusion horizon.
\end{theorem}

\begin{proof}
\textit{Single-node case.} Fix $A\subseteq B\subseteq V$,
$v\in V\setminus B$. Define
\begin{equation}
  g_i(x) := \big(\max\{0,b_i-x\}\big)^p.
\end{equation}
$g_i$ is a composition of the non-increasing convex map
$x\mapsto\max\{0,b_i-x\}$ with the non-decreasing convex map
$z\mapsto z^p$ ($p\ge1$), hence convex and non-increasing.

 Convexity implies: for fixed $\delta\ge0$,
\begin{equation}
  g_i(x)-g_i(x+\delta)
\end{equation}
is non-increasing in $x$.

 By Lemma~\ref{lem:prelim}, $x_i^A(T)\le x_i^B(T)$
and
\begin{equation}
  \delta_i^A(T) \ge (\gamma^+)^T \delta_i^B(T).
\end{equation}
Applying Step 2 at base point $x_i^A(T)$ with step $\delta_i^A(T)$,
\begin{multline}
  g_i\big(x_i^A(T)\big) - g_i\big(x_i^{A\cup\{v\}}(T)\big) \\
  \;\ge\; (\gamma^+)^T
  \Big[g_i\big(x_i^B(T)\big) - g_i\big(x_i^{B\cup\{v\}}(T)\big)\Big].
  \label{eq:single-node-hm}
\end{multline}

 Multiplying \eqref{eq:single-node-hm} by
$w_i\ge0$ and summing over $i\in V$,
\begin{multline}
  R_p(A\cup\{v\}) - R_p(A) \\
  \;\ge\; (\gamma^+)^T \big[R_p(B\cup\{v\}) - R_p(B)\big].
  \label{eq:single-node-Rp}
\end{multline}

\textit{Telescoping.} Let $U=\{u_1,\dots,u_m\}$ and
\begin{equation}
  B_0 := B, \qquad B_k := B\cup\{u_1,\dots,u_k\}.
\end{equation}
Then
\begin{equation}
  R_p(B\cup U) - R_p(B)
  = \sum_{k=1}^{m} \big[R_p(B_k) - R_p(B_{k-1})\big].
  \label{eq:telescope}
\end{equation}

Since $A\subseteq B_{k-1}$ and
$u_k\notin B_{k-1}$, \eqref{eq:single-node-Rp} applies with
$(A,B_{k-1},u_k)$ in place of $(A,B,v)$:
\begin{multline}
  R_p(A\cup\{u_k\}) - R_p(A) \\
  \;\ge\; (\gamma^+)^T \big[R_p(B_{k-1}\cup\{u_k\}) - R_p(B_{k-1})\big].
  \label{eq:per-term}
\end{multline}

 Summing \eqref{eq:per-term} over
$k=1,\dots,m$ and using \eqref{eq:telescope},
\begin{multline}
  \sum_{k=1}^{m}\big[R_p(A\cup\{u_k\}) - R_p(A)\big] \\
  \;\ge\; (\gamma^+)^T\big[R_p(B\cup U) - R_p(B)\big],
\end{multline}
which is \eqref{eq:weaksub-def}. Positivity of $(\gamma^+)^T$ follows
from $\alpha_{\min}>0$ and $L^->0$ (Remark~\ref{rem:gate-positivity}),
with $T<\infty$.
\hfill$\blacksquare$
\end{proof}

\subsubsection{NP-Hardness}

The preceding properties guarantee approximation through greedy
optimization, but exact optimization remains computationally hard.

\begin{theorem}[NP-Hardness of Harm Minimization]
\label{thm:np-hard}
With $H_1(A) := \sum_{i \in V} w_i\, s_i(A)$ as in
Proposition~\ref{prop:mono-hm} (the case $p=1$),\footnote{The
restriction to $p=1$ is necessary, not just convenient: the map
$z\mapsto z^p$ is pointwise monotone but does not, in general,
preserve the $H_1$-ordering of seed sets across $p$. For instance,
shortfall vectors $(0.6,0.6)$ and $(0,1)$ give $H_1 = 1.2$ vs.\ $1.0$,
but $H_2 = 0.72$ vs.\ $1.0$---so the minimizer at $p=1$ differs from
the minimizer at $p=2$, and NP-hardness at $p=1$ does not immediately
transfer to $p>1$. Hardness for $p>1$ is left open.} the problem of
finding $\arg\min_{|A|\le k} H_1(A)$ is NP-hard.
\end{theorem}

\begin{proof}
Reduce from Influence Maximization (IM) under the Independent Cascade
(IC) model, which is NP-hard through classical reductions from
\textsc{Set Cover}. Given an IC instance on graph $G'$, set
\begin{equation*}
  E^-=\emptyset, \qquad
  \lambda_i=0, \qquad
  b_i=1, \qquad
  w_i=1.
\end{equation*}
for all $i$, and $p=1$. Under this regime, RASH's diffusion
coincides with the IC process on $G'$.

By Lemma~\ref{lem:prelim} (a), $x_i^A(T)\in[0,1]$, so
\begin{equation*}
  s_i(A) = \max\{0,1-x_i^A(T)\} = 1-x_i^A(T).
\end{equation*}

 Substituting,
\begin{equation*}
  H_1(A)
  = \sum_{i\in V}\big(1-x_i^A(T)\big)
  = |V| - \sum_{i\in V}x_i^A(T).
\end{equation*}

 Since $|V|$ is constant in $A$,
\begin{equation*}
  \arg\min_{|A|\leq k} H_1(A)
  =
  \arg\max_{|A|\leq k}
  \sum_{i\in V} x_i^A(T).
\end{equation*}
The right-hand side is exactly the Influence Maximization (IM)
objective under the Independent Cascade (IC) model, which seeks to
maximize the expected number of influenced nodes using at most $k$
seeds. This problem is NP-hard~\cite{4a}. Thus, minimizing $H_1(A)$ is
equivalent to solving the NP-hard IM problem.

 A polynomial-time algorithm for HM would solve
IM in polynomial time. Since IM under IC is NP-hard, so is
minimizing $H_1$.

\end{proof}

\subsection{Seed Selection}
\label{sec:seed-selection}

The theoretical results established above are only useful if they
translate into a concrete, implementable procedure for selecting seed
nodes. Theorem~\ref{thm:np-hard} rules out an efficient exact solver, so
this translation must take the form of an approximation algorithm with a
provable guarantee, rather than an exact one. Fortunately,
Proposition~\ref{prop:mono-hm} and Theorem~\ref{thm:weaksub-hm} show that
the harm-reduction function $R_p$ is monotone and $(\gamma^+)^T$-weakly
submodular, exactly the two properties a marginal-gain greedy procedure
needs to admit a provable approximation guarantee. This section develops
that procedure - we give a greedy algorithm, show that its approximation
guarantee follows directly from Proposition~\ref{prop:mono-hm} and
Theorem~\ref{thm:weaksub-hm}, and close by characterizing its
computational cost.

We fix two shorthands used throughout this section:
$\textsc{Rash}(A)$ denotes one deterministic RASH simulation from seed
set $A$ (Algorithm~\ref{alg:rash}), returning the terminal state vector
$\mathbf{x}^A(T)$ and
\begin{equation*}
\mathrm{Harm}(\mathbf{x})
:=
\sum_{i \in V}
w_i \bigl(\max\{0, b_i - x_i\}\bigr)^p
\end{equation*}
gives the harm $H_p(A) = \mathrm{Harm}\bigl(\textsc{Rash}(A)\bigr)$ of
the resulting state, for any state vector $\mathbf{x}$. Because RASH's
update rule is deterministic given $\{\lambda_i,\theta_i,\alpha_i\}$,
$\textsc{Rash}(A)$ requires no repeated sampling to evaluate---a single
rollout per seed set suffices, in contrast to simulation-based signed IM
methods, whose marginal gains must be averaged over many stochastic
cascade realizations. Algorithm~\ref{alg:hm} builds on these two
primitives to select a harm-minimizing seed set.

\begin{algorithm}[t]
\caption{RASH}
\label{alg:rash}
\begin{algorithmic}[1]
\Require $G=(V,E^+,E^-,a)$, seed set $A$, horizon $T$,
parameters $\{\lambda_i,\theta_i,\alpha_i\}_{i\in V}$,
in-neighborhoods $\{\Gamma_i^+,\Gamma_i^-\}_{i \in V}$
\Ensure Terminal state $\mathbf{x}^A(T)$
\For{each $i \in V$}
    \If{$i \in A$}
        \State $x_i^A(0) \leftarrow 1$ \Comment{seeded node}
    \Else
        \State $x_i^A(0) \leftarrow 0$ \Comment{unseeded node}
    \EndIf
\EndFor
\For{$t = 0$ \textbf{to} $T-1$}
    \For{each $i \in V$}
        \State $P_i(t) \leftarrow \sum_{j \in \Gamma_i^+} a_{ji}\,x_j^A(t)$ \Comment{supportive input}
        \State $N_i(t) \leftarrow \sum_{j \in \Gamma_i^-} a_{ji}\,x_j^A(t)$ \Comment{discouraging input}
        \State $I_i(t) \leftarrow (1-\lambda_i)\,P_i(t) + \lambda_i\,N_i(t)$ \Comment{susceptibility-weighted combination}
        \State $\tilde I_i(t) \leftarrow \tanh\big(I_i(t)\big)$ \Comment{normalization}
        \If{$|\tilde I_i(t)| \geq \theta_i$}
            \State $x_i^A(t{+}1) \leftarrow (1-\alpha_i)\,x_i^A(t) + \alpha_i\,\tilde I_i(t)$ \Comment{threshold cleared}
        \Else
            \State $x_i^A(t{+}1) \leftarrow (1-\alpha_i)\,x_i^A(t)$ \Comment{decay only}
        \EndIf
    \EndFor
\EndFor
\State $\mathbf{x}^A(T) \leftarrow \big(x_1^A(T), x_2^A(T), \ldots, x_{|V|}^A(T)\big)$ \Comment{terminal awareness vector}
\State \Return $\mathbf{x}^A(T)$
\end{algorithmic}
\end{algorithm}

\begin{algorithm}[t]
\caption{Harm Minimization (Greedy)}
\label{alg:hm}
\begin{algorithmic}[1]
\Require $G=(V,E^+,E^-,a)$, budget $k$, horizon $T$,
$\{b_i\}_{i\in V}$, $\{w_i\}_{i\in V}$, $p \geq 1$
\Ensure Seed set $A$, $|A| = k$
\State $A \leftarrow \emptyset$ \Comment{no seeds placed yet}
\State $\mathbf{x} \leftarrow \textsc{Rash}(A)$ \Comment{uninformed-population baseline}
\State $H \leftarrow \mathrm{Harm}(\mathbf{x})$ \Comment{$H = H_p(\emptyset)$}
\For{$j = 1$ \textbf{to} $k$}
    \State $\Delta^\star \leftarrow -\infty$; \; $v^\star \leftarrow \text{null}$
    \For{each candidate $v \in V \setminus A$}
        \State $\mathbf{x}_v \leftarrow \textsc{Rash}(A \cup \{v\})$ \Comment{trial rollout}
        \State $H_v \leftarrow \mathrm{Harm}(\mathbf{x}_v)$ \Comment{harm if $v$ is added}
        \State $\Delta(v) \leftarrow H - H_v$ \Comment{marginal harm reduction}
        \If{$\Delta(v) > \Delta^\star$}
            \State $\Delta^\star \leftarrow \Delta(v)$; \; $v^\star \leftarrow v$
        \EndIf
    \EndFor
    \State $A \leftarrow A \cup \{v^\star\}$ \Comment{commit best candidate}
    \State $H \leftarrow H - \Delta^\star$ \Comment{update running harm, no re-rollout}
\EndFor
\State \Return $A$
\end{algorithmic}
\end{algorithm}

\paragraph{Harm Minimization algorithm.}
Algorithm~\ref{alg:hm} builds the seed set one node at a time.
Lines~1--3 establish the baseline harm $H_p(\emptyset)$ by rolling out
RASH with no seeds placed. Each outer iteration then evaluates every
remaining candidate by simulating its trial addition to the current
seed set, scoring the resulting harm, and recording the marginal
reduction $\Delta(v)$ relative to the seed set's \emph{current} state
rather than the empty set---this is what lets the algorithm implicitly
account for interaction effects between previously chosen seeds and
each new candidate. The best candidate is committed, and the running
harm is updated by subtraction rather than by a redundant rollout.
Ties may be broken arbitrarily without affecting the guarantee below,
since Corollary~\ref{cor:greedy-approx} holds for any valid greedy
selection at each step.

\paragraph{Approximation guarantee}

With monotonicity and $\gamma$-weak submodularity of $R_p$ established, the greedy
algorithm's performance can now be bounded relative to the optimal seed
set, following the standard argument for greedy maximization under
weak submodularity.

\begin{corollary}[Approximation Ratio of Harm Minimization]
\label{cor:greedy-approx}
Let $A_{\mathrm{greedy}}$ be the set returned by
Algorithm~\ref{alg:hm} and
$A^\star = \arg\max_{|A|\le k} R_p(A)$ an optimal solution.
Then
\begin{equation}
R_p(A_{\mathrm{greedy}}) \;\geq\; \bigl(1 - e^{-(\gamma^+)^T}\bigr)\, R_p(A^\star).
\label{eq:greedy-bound}
\end{equation}
\end{corollary}

\begin{IEEEproof}
Let $\gamma := (\gamma^+)^T$, let $A_j$ denote the greedy set after $j$
iterations, and let $\Delta_j := R_p(A^\star) - R_p(A_j)$. Applying
Theorem~\ref{thm:weaksub-hm} with $B = A_j$, $U = A^\star \setminus
A_j$, together with monotonicity (Proposition~\ref{prop:mono-hm}),
\begin{equation}
\sum_{v \in A^\star \setminus A_j} \bigl[R_p(A_j \cup \{v\}) - R_p(A_j)\bigr] \geq \gamma \Delta_j.
\end{equation}
Since $|A^\star \setminus A_j| \leq k$ and greedy takes the largest
single term on the left,
\begin{equation}
R_p(A_{j+1}) - R_p(A_j) \;\geq\; \frac{\gamma}{k}\, \Delta_j.
\end{equation}
This gives $\Delta_{j+1} \leq (1 - \gamma/k)\, \Delta_j$, so with
$\Delta_0 = R_p(A^\star)$ and $(1-\gamma/k)^k \leq e^{-\gamma}$,
\begin{equation}
\Delta_k \;\leq\; e^{-\gamma} R_p(A^\star),
\end{equation}
which rearranges to \eqref{eq:greedy-bound}.
\end{IEEEproof}

Equivalently, $H_p(A_{\mathrm{greedy}}) \leq H_p(A^\star) +
e^{-(\gamma^+)^T}\bigl[H_p(\emptyset) - H_p(A^\star)\bigr]$: Harm
Minimization's output never exceeds the optimum by more than an
$e^{-(\gamma^+)^T}$-fraction of the maximum achievable reduction.

\paragraph{Complexity.}
Each call to Algorithm~\ref{alg:rash} costs $O(T|E|)$, and each outer
iteration of Algorithm~\ref{alg:hm} re-evaluates every remaining
candidate from scratch, so building a full seed set of size $k$ costs
$O(k\,|V|\,T\,|E|)$. The dominant cost is the number of full
$\textsc{Rash}$ rollouts performed: $k$ outer iterations, each
re-evaluating up to $|V|$ remaining candidates, with each evaluation
requiring one $O(T|E|)$ rollout.

CELF \cite{12a} returns exactly the same seed set as
Algorithm~\ref{alg:hm} whenever the marginal gains it re-scores are
true, monotonically non-increasing upper bounds across iterations. Under
$(\gamma^+)^T$-weak submodularity, this monotone-bound property is not
guaranteed to hold for every instance, so we use
CELF here as an empirically effective
acceleration of Algorithm~\ref{alg:hm} rather than as a formally
guaranteed equivalent.

\section{Experiments}

The theoretical results in Section~\ref{sec:hm-props} establish that RASH and Harm Minimization are tractable optimization frameworks with monotone, weakly submodular objectives and that greedy seed selection provides a provable approximation guarantee. However, these results alone do not show whether the phenomena motivating RASH---awareness reversal, suppression under discouraging influence, and disproportionate harm to vulnerable individuals---arise in real signed networks, or whether Harm Minimization provides practical benefits when seed selection is compared with simpler heuristics under computational constraints.

We address these questions in two phases. Phase~1 evaluates whether RASH captures diffusion behaviors that existing signed diffusion models cannot, by comparing its propagation dynamics with four representative signed-network baselines. Phase~2 evaluates whether Harm Minimization, using the greedy seed-selection strategy from Section~\ref{sec:seed-selection}, protects vulnerable individuals more effectively than Influence Maximization, Influence Minimization, and several practical seed-selection heuristics, while also examining the associated computational cost. We first describe the experimental setup common to both phases, followed by the results.
\subsection{Experimental Setup}

\subsubsection{Datasets and Baseline Models}
We evaluate RASH and the Harm Minimization (HM) objective on
one synthetic and five real-world signed networks widely used
in prior signed-network research: a synthetic Barab\'asi--Albert
(BA) graph, Bitcoin Alpha, Bitcoin OTC, Epinions, Slashdot, and
Wiki-RfA. The five real-world networks are obtained from the
Stanford Large Network Dataset Collection
(SNAP)\footnote{\url{https://snap.stanford.edu/data/index.html}}.
Table~\ref{tab:datasets} summarizes their statistics, followed
by a brief description of each.

\begin{table}[H]
\caption{Summary statistics of the signed social network datasets used for evaluation.}
\label{tab:datasets}
\centering
\setlength{\tabcolsep}{4pt}
\renewcommand{\arraystretch}{1.05}
\begin{tabular}{lrrcc}
\hline
\textbf{Dataset} & \textbf{Nodes} & \textbf{Edges} & \textbf{Weighted} & \textbf{Temporal} \\
\hline
BA (synthetic) & 500 & -- & No & No \\
Bitcoin Alpha & 3,783 & 24,186 & Yes & Yes \\
Bitcoin OTC & 5,881 & 35,592 & Yes & Yes \\
Epinions & 131,828 & 841,372 & No & No \\
Slashdot & 81,871 & 545,671 & No & No \\
Wiki-RfA & 10,835 & 159,388 & No & No \\
\hline
\end{tabular}
\end{table}

1) \textit{BA (synthetic)}: A Barab\'asi--Albert scale-free
graph with $N=500$ nodes, generated with preferential
attachment and used as a controlled setting for validating
RASH's dynamics before evaluating on real signed networks.
Edge signs are assigned independently with positive-edge
probability $p^+ \in \{0.2, 0.5, 0.8\}$ to probe RASH's
sensitivity to the underlying sign structure.

2) \textit{Bitcoin Alpha / Bitcoin OTC}: Who-trusts-whom
networks in which users of two Bitcoin trading platforms rate
one another's trustworthiness on a scale of $-10$ to $+10$; we
take the sign of each rating as the edge polarity and its
magnitude as the edge weight $a_{ji}$.

3) \textit{Epinions}: A trust/distrust network from the
Epinions product-review platform, where users designate other
reviewers as trusted or distrusted; among our five real-world
datasets, it is the largest and thus the primary stress test for
scalability.

4) \textit{Slashdot}: The Slashdot Zoo signed social network, in
which users tag one another as ``friend'' or ``foe,'' giving a
polarized signed graph substantially denser in negative ties
than the trust-rating networks above.

5) \textit{Wiki-RfA}: Wikipedia Requests for Adminship voting
data, where each support or oppose vote cast by one editor on
another's candidacy is treated as a positive or negative directed
edge, respectively.

We compare RASH and HM against two groups of baselines,
corresponding to the two experimental tracks below.

\textit{Diffusion-dynamics baselines.} To validate that RASH
captures signed, non-monotonic diffusion phenomena that existing
signed diffusion models cannot, we compare against the four known
representative models as described below:

 (i) \textit{Personal Information Diffusion Model (PID) }\cite{PID}: Models information spread through polarity-dependent node states, with messages transitioning among unknown, insider, publisher, and forgetter states

 (ii) \textit{Signed Linear Threshold Model (SLT) }\cite{LTE}: Separately accumulates positive and negative influence, activating a node when their net signed influence exceeds its threshold.

 (iii) \textit{Signed Independent cascade Model (SNIC) }\cite{ICE}:Allows active nodes to probabilistically activate neighbors with polarity determined by the connecting signed edge, enabling competing positive and negative activation.

 (iv) \textit{Polarity Related Linear Influence Diffusion (PLID) }\cite{PLID}: Propagates positive and negative influence probabilities through linear iterations, avoiding Monte Carlo simulation and discrete activation.
 
 All four treat activation as a discrete,
irreversible event and are re-implemented under identical
network inputs and seed sets to ensure a fair comparison.

\textit{Harm Minimization (HM) baselines,} This optimzation objective being novel does not have any baseline, therefore we compare it with influence maximization and influence minimzation objectives in which greedy algorithm technique is being used for seed selection under RASH Model. We try to observe if our objective can maximize the influence through seeds while taking care of vulnerable nodes and hence reducing the harm. To evaluate the HM objective against the running time we have used different strategies to select the seeds and also checked the harm reduced at the cost of time employed> Following strategies are employed: 
(i) \textit{Random}: Selects the $k$ seed nodes uniformly at random, providing a simple and computationally inexpensive baseline for seed selection.

(ii) \textit{High Degree} \cite{HighDegree}:, which greedily
selects seeds to maximize direct coverage of a designated
vulnerable-node set, without accounting for shortfall severity
or vulnerability weighting; 

(iii) \textit{Dynamic Degree Discount} \cite{DynamicDegree}: Extends Degree Discount to temporal networks by using dynamic node degrees that capture evolving neighborhoods and discounting the scores of nodes adjacent to already selected seeds across time.

(iv) \textit{Betweenness}\cite{HighDegree}:Selects nodes with the highest betweenness centrality, prioritizing nodes that frequently lie on shortest paths and can therefore serve as important bridges for information propagation.

(v) \textit{Targeted Protection} \cite{TargetedProtection}: Prioritizes seed/protector nodes according to their ability to protect a specified target set from undesirable influence.

(vi) \textit{Reactive} \cite{reactive}: Dynamically places protective/trusted seeds after the locations of misinformation or negative influence become known, recomputing the seed positions for each spreading instance.

\subsubsection{Implementation Details}
All experiments are conducted on two NVIDIA A100 GPUs (80~GB
each). For the synthetic BA network, node-level parameters are
sampled as $\alpha_i \sim U(0.2, 0.8)$, $\lambda_i \sim U(0,1)$,
and $\theta_i \sim U(0.1, 0.4)$, following Section~IV-A. For the
five real-world networks, edge signs and weights are taken
directly from the data (Bitcoin ratings are normalized to
$[0,1]$ by magnitude; Slashdot, Epinions, and Wiki-RfA edges are
treated as unweighted, $a_{ji}=1$); node parameters are instead
computed from network structure using the estimators of
Section~IV-A, $\lambda_i = |\Gamma_i^-|/(|\Gamma_i^+| +
|\Gamma_i^-|)$ and $\alpha_i$ as the ratio of a node's in-degree
to the maximum observed in-degree, while $\theta_i$ is sampled
uniformly as in the synthetic setting, since no natural
resistance threshold is present in the raw data. All experiments
use a diffusion horizon of $T=60$. Seeds are initialized with
$x_i(0)=1$ for the positive regime and blockers with
$x_i(0)=-1$ for the negative regime, per Lemmas~1 and~3.

Because RASH's propagation rule is deterministic given a fixed
seed set (Section~IV-A), repeated trials are needed only to
average over randomness in baseline seed \emph{selection} (e.g.,
tie-breaking among equally central candidates), not over
stochastic cascades; each configuration is averaged over 20
independent trials.

\subsection{Results}

We report results in two phases. Phase~1 evaluates 
diffusion dynamics of RASH against the four signed baselines (PID, SLT,
SNIC, PLID) on synthetic and Epinions dataset. Phase~2, presented separately  will evaluate the Harm Minimization objective against its own baselines under the datasets given in Table \ref{tab:datasets}.

\subsubsection{Phase 1 RASH Diffusion Dynamics}
Fig.~\ref{rashvalidity} reports RASH's diffusion behavior against
four baseline models (PID, SLT, SNIC, PLID) across all six
datasets: BA (panel~a), Bitcoin Alpha (b), Bitcoin OTC (c),
Epinions (d), Slashdot (e), and Wiki-RfA (f). For each dataset,
three metrics are shown: (i) the cumulative count of reversed
nodes -- nodes that were positively aware at some time $t_1$ and
negatively aware at a later time $t_2 > t_1$; (ii) the cumulative
count of suppression events -- time steps at which a node's
awareness decreases despite having been positive at the previous
step; and (iii) negative spread $\sigma^-(S)$ over time under
greedy blocking budgets $k \in \{0,5,10,15\}$.

\paragraph{Reversed nodes and suppression events}
RASH is the only model that ever records a nonzero reversal or
suppression count; PID, SLT, SNIC, and PLID remain flat at zero
across all six networks, from $500$-node BA to $131$k-node
Epinions -- a categorical distinction, since SLT and SNIC enforce
irreversible activation by construction and PID/PLID lack a
suppression mechanism entirely. RASH's curves rise and plateau
within $10$--$20$ steps on the smaller networks (BA, Bitcoin
Alpha, Bitcoin OTC), while on the larger, denser networks
(Epinions, Slashdot, Wiki-RfA) suppression events keep
accumulating across the full horizon, consistent with these
networks offering more negative-edge pathways for eroding prior
awareness.

\paragraph{Negative spread under blocking}
Panel~(c) of each dataset shows a clear, scale-dependent pattern.
On the three smaller networks, increasing $k$ produces
well-separated, rapidly stabilizing plateaus, with the largest
reduction between $k{=}0$ and $k{=}5$ and diminishing reductions
thereafter,  the empirical signature of the $(\gamma^-)^T$-weak
submodularity guarantee. On the three larger
networks, negative spread instead oscillates persistently under
all budgets, with blocking budgets separated only within the
oscillation band. This is expected, not a failure: $\gamma^- =
\alpha_{\min}\cdot\mathrm{sech}^2(R^-)$ depends on $R^- =
M\alpha^-_{\max}|V|$, which grows with $|V|$, and
$\mathrm{sech}^2(\cdot)$ decays monotonically, so larger
networks admit a strictly weaker worst-case guarantee at the same
budget. Epinions, Slashdot, and Wiki-RfA are one to two orders of
magnitude larger than BA, Bitcoin Alpha, and Bitcoin OTC, so the
added volatility is precisely what a shrinking $(\gamma^-)^T$
predicts.

\begin{figure*}[t]
\centering

% Row 1: (a), (b)
\subfloat[BA (synthetic)]{%
    \includegraphics[width=0.47\textwidth]{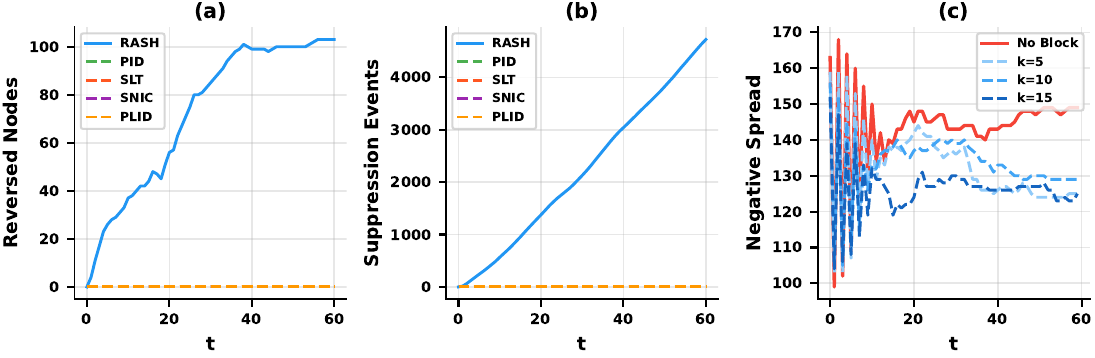}%
    \label{fig:activated_ba}
}
\hfill
\subfloat[Bitcoin Alpha]{%
    \includegraphics[width=0.47\textwidth]{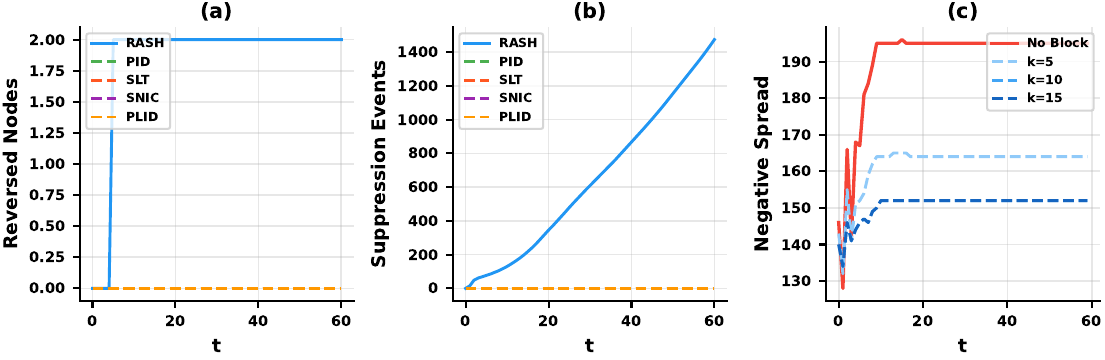}%
    \label{fig:activated_alpha}
}

\vspace{0.8em}

% Row 2: (c), (d)
\subfloat[Bitcoin OTC]{%
    \includegraphics[width=0.47\textwidth]{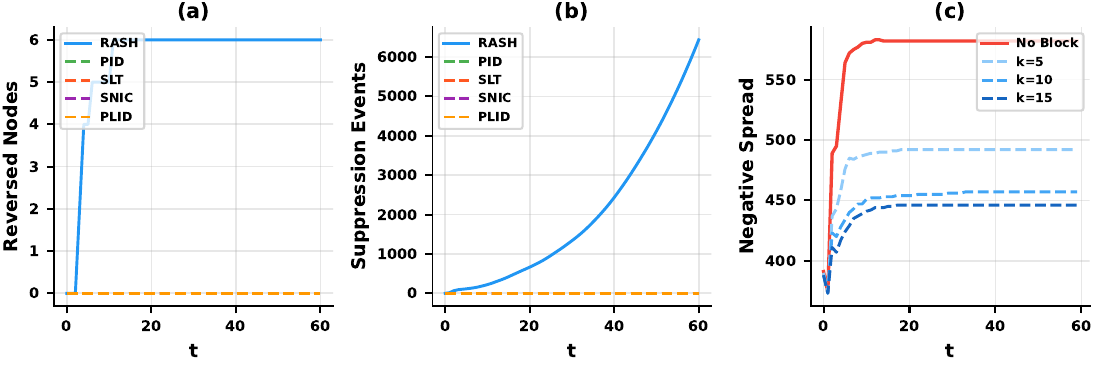}%
    \label{fig:activated_otc}
}
\hfill
\subfloat[Epinions]{%
    \includegraphics[width=0.47\textwidth]{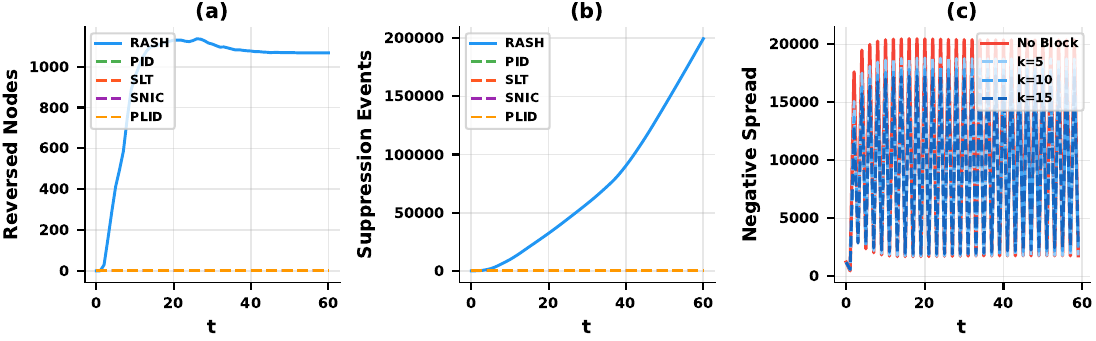}%
    \label{fig:activated_epinions}
}

\vspace{0.8em}

% Row 3: (e), (f)
\subfloat[Slashdot]{%
    \includegraphics[width=0.47\textwidth]{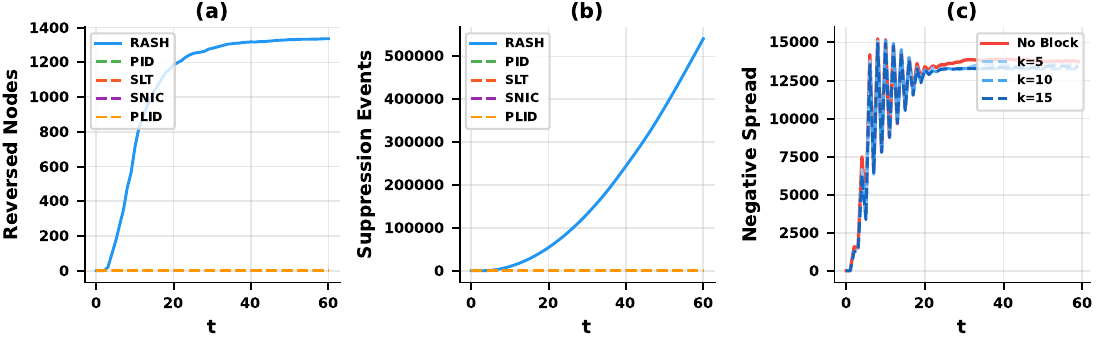}%
    \label{fig:activated_slashdot}
}
\hfill
\subfloat[Wiki-RfA]{%
    \includegraphics[width=0.47\textwidth]{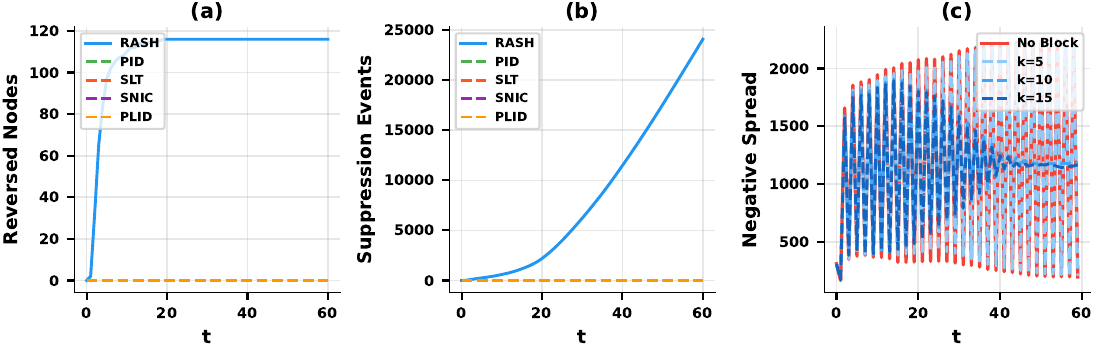}%
    \label{fig:activated_wikirfa}
}

\caption{RASH versus baseline diffusion models across various datasets.
For each dataset, (a) shows the cumulative count of nodes that were
positively aware at some time $t_1$ and negatively aware at a later time
$t_2>t_1$; (b) shows the cumulative count of suppression events, where a
node's awareness decreases despite having been positive at the previous
step; and (c) shows the negative spread $\sigma^-(S)$ over time under
greedy blocking budgets $k\in\{0,5,10,15\}$.}
\label{rashvalidity}
\end{figure*}

\subsubsection{Phase 2 Harm Minimization}
This phase evaluates the Harm Minimization (HM) objective directly
against its two theoretical boundary cases, Influence Maximization
(IM) and Influence Minimization (Inf-Min), and against six seed-selection
heuristics adapted from the influence-maximization and minimization literature \cite{HighDegree} \cite{TargetedProtection} \cite{reactive} \cite{DynamicDegree} .
Fig.~\ref{fig:activated_all} reports the number of activated nodes
against seed set size for IM, Inf-Min, and HM across all six
datasets; Fig.~\ref{fig:time_harm} reports the corresponding
percentage harm reduction; and Table~\ref{tab:runtime_comparison} reports
seed-selection running time alongside harm reduction for HM against
the  heuristic baselines 

\paragraph{Activated nodes across objectives.}
Fig.~\ref{fig:activated_all} shows Inf-Min's curve remaining near
zero across every dataset and seed budget, confirming that a purely
suppressive objective, by construction, never seeks positive reach.
IM's curve rises sharply and dominates throughout, as expected of an
objective whose sole aim is maximizing raw activation. HM's curve
consistently tracks between the two, and its position relative to IM
is itself informative: on the smaller, sparser networks (BA, Bitcoin
Alpha, Bitcoin OTC, Wiki-RfA), HM closely trails IM, activating
within a few percent of IM's reach at every budget; on the larger,
denser networks (Epinions, Slashdot), the gap between HM and IM
widens substantially. This is consistent with harm-aware seeding
redirecting a larger share of the budget toward protecting vulnerable
individuals precisely where the network offers more opportunities for
harm, exactly the networks in which Phase~1 (Fig.~\ref{rashvalidity})
showed the most persistent negative-spread activity. HM's reduced
activation on these networks is therefore not a shortcoming but a
deliberate consequence of the objective: it declines to chase
additional raw reach once that reach no longer contributes to
reducing weighted shortfall.

\begin{figure*}[t!]
\centering

\subfloat[BA (synthetic)]{%
    \includegraphics[width=0.32\textwidth]{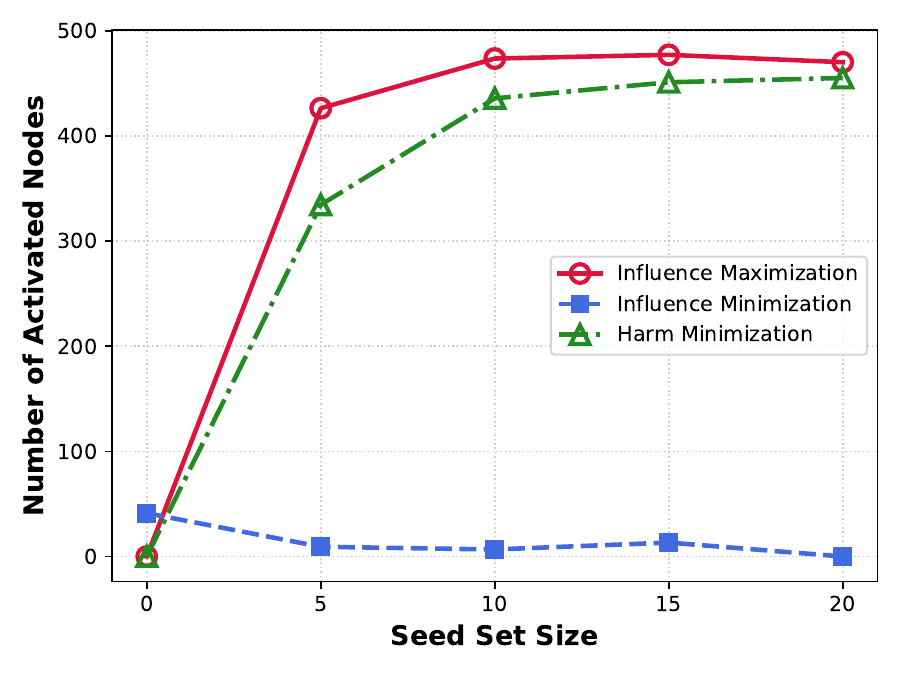}%
    \label{fig:comp_ba}
}
\hfill
\subfloat[Bitcoin Alpha]{%
    \includegraphics[width=0.32\textwidth]{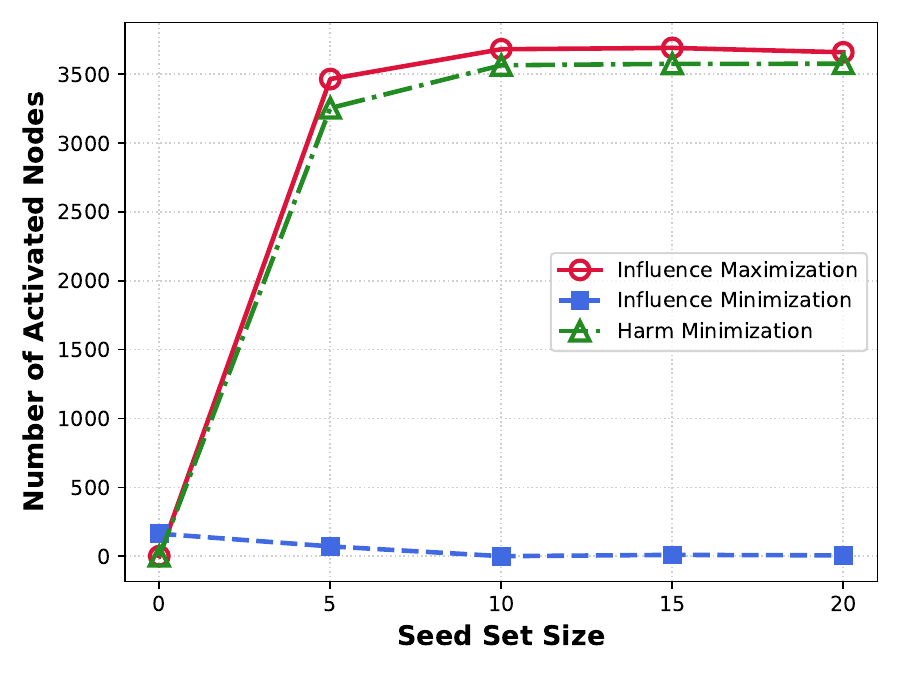}%
    \label{fig:comp_bita}
}
\hfill
\subfloat[Bitcoin OTC]{%
    \includegraphics[width=0.32\textwidth]{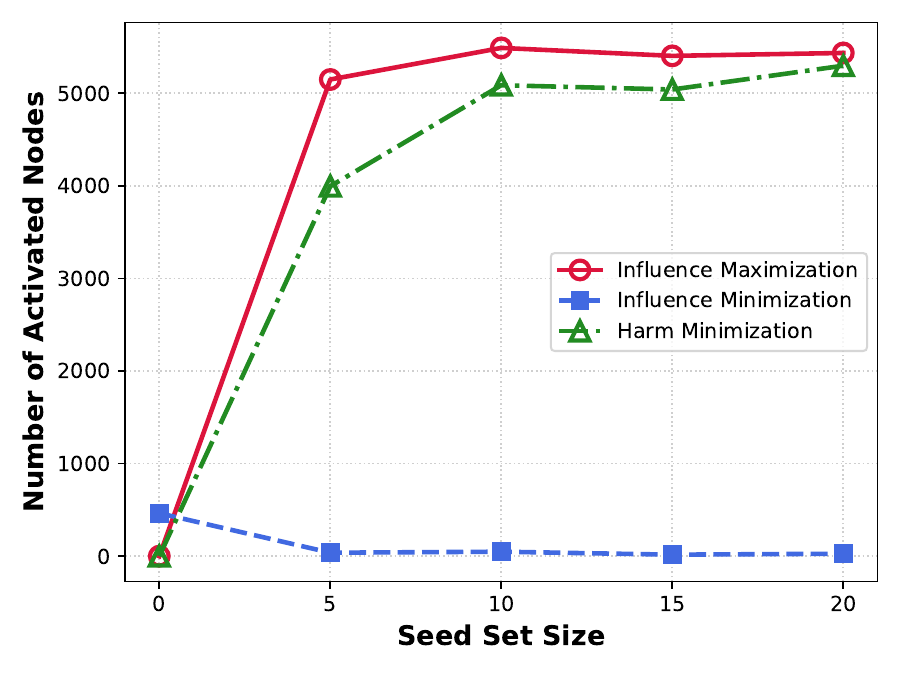}%
    \label{fig:comp_otc}
}

\vspace{0.6em}

\subfloat[Epinions]{%
    \includegraphics[width=0.32\textwidth]{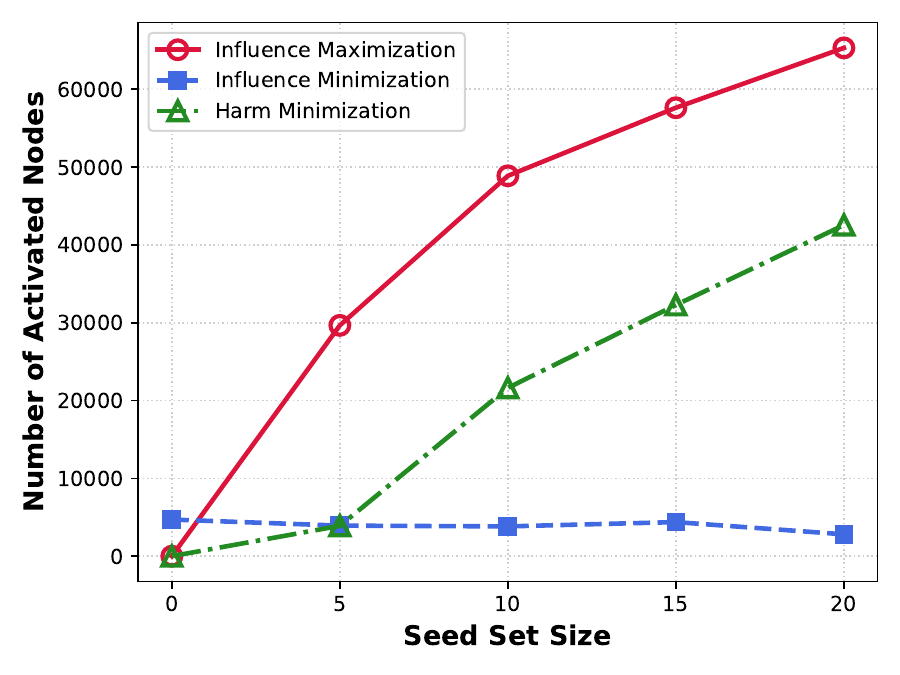}%
    \label{fig:comp_epinions}
}
\hfill
\subfloat[Slashdot]{%
    \includegraphics[width=0.32\textwidth]{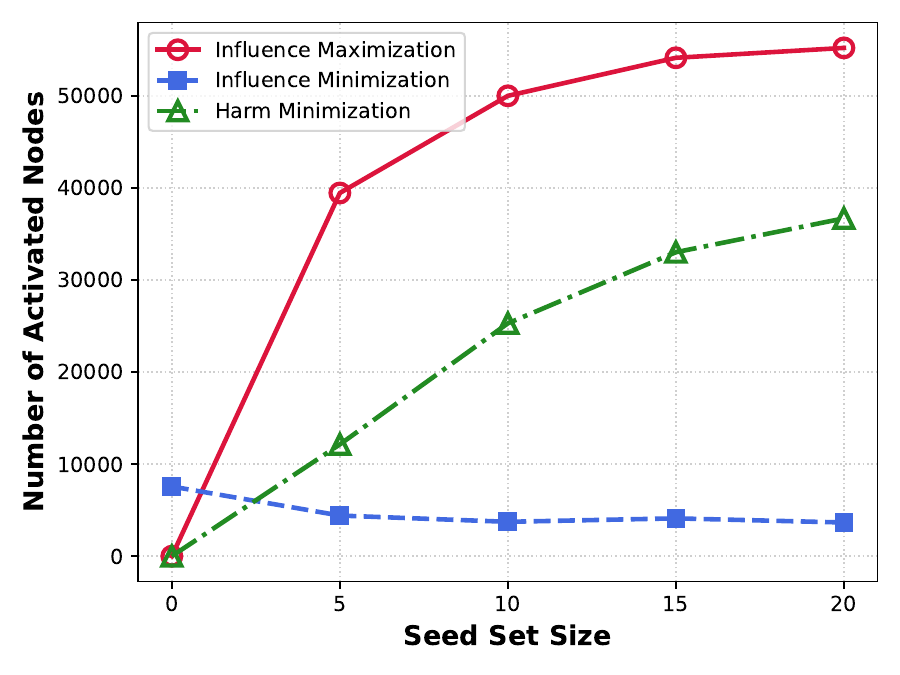}%
    \label{fig:comp_slashdot}
}
\hfill
\subfloat[Wiki-RfA]{%
    \includegraphics[width=0.32\textwidth]{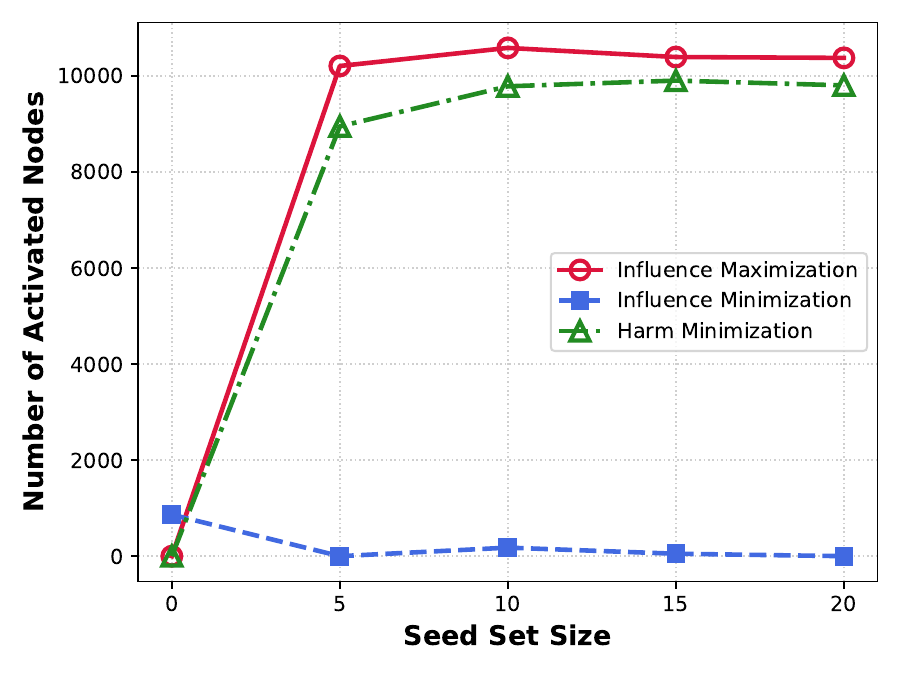}%
    \label{fig:comp_wikirfa}
}

\caption{Number of activated nodes versus seed set size for Influence
Maximization, Influence Minimization, and Harm Minimization across all
six datasets.}
\label{fig:activated_all}
\end{figure*}

\paragraph{Harm reduction across objectives}
Fig.~\ref{fig:time_harm} makes the practical consequence of this
design explicit. HM achieves the highest percentage harm reduction at
every seed budget on every dataset, typically approaching or reaching
full saturation on the smaller networks (BA, Bitcoin Alpha, Bitcoin
OTC) by $k{=}10$--$15$, and maintaining a decisive margin over both
IM and Inf-Min on the larger networks even where saturation is not
reached. IM, despite dominating raw activation in
Fig.~\ref{fig:activated_all}, plateaus at a substantially lower harm
reduction ($\sim$50--65\% on most datasets), underscoring that
maximizing spread and minimizing harm are not the same objective:
IM routes influence toward whichever nodes are structurally easiest
to reach, with no mechanism to preferentially protect individuals
who need it most. Inf-Min performs weakest of the three throughout,
since suppressing negative spread alone provides no positive
awareness to the population it is meant to protect. This pattern
holds without exception across all six datasets, indicating that
HM's advantage is a property of the objective itself rather than an
artifact of any single network's structure or scale.

\begin{figure*}[t!]
\centering

\subfloat[BA (synthetic)]{%
    \includegraphics[width=0.32\textwidth]{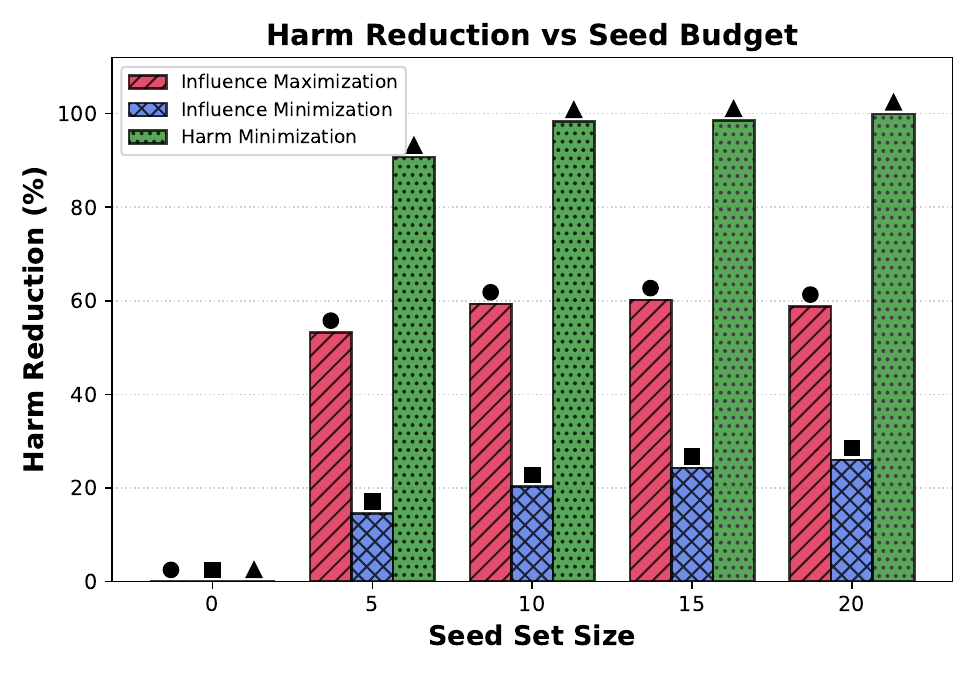}%
    \label{fig:harm_ba}
}
\hfill
\subfloat[Bitcoin Alpha]{%
    \includegraphics[width=0.32\textwidth]{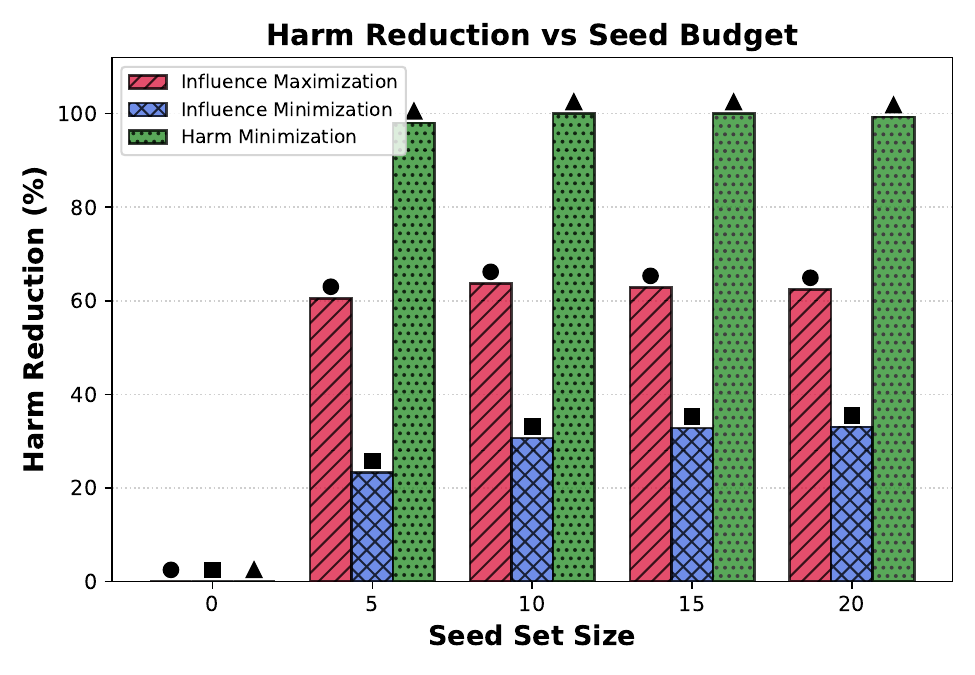}%
    \label{fig:harm_alpha}
}
\hfill
\subfloat[Bitcoin OTC]{%
    \includegraphics[width=0.32\textwidth]{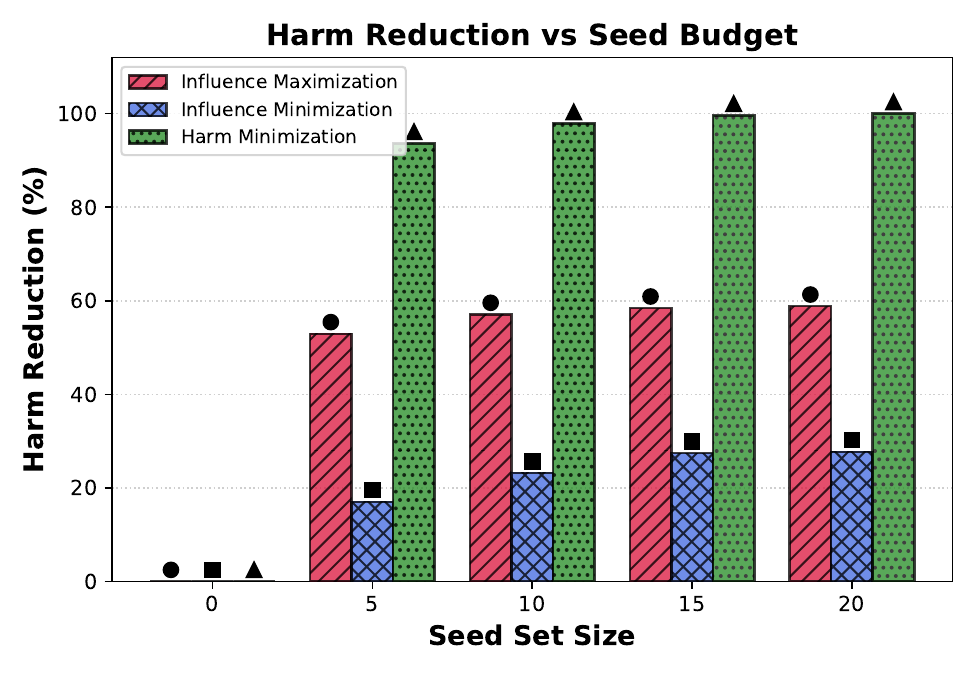}%
    \label{fig:harm_otc}
}

\vspace{0.6em}

\subfloat[Epinions]{%
    \includegraphics[width=0.32\textwidth]{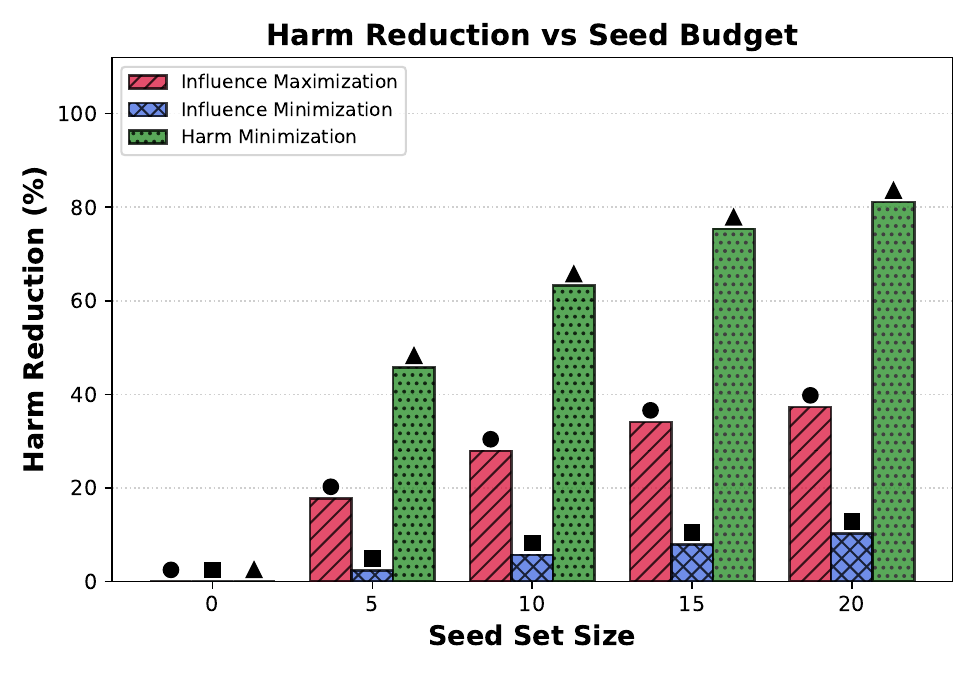}%
    \label{fig:harm_epinions}
}
\hfill
\subfloat[Slashdot]{%
    \includegraphics[width=0.32\textwidth]{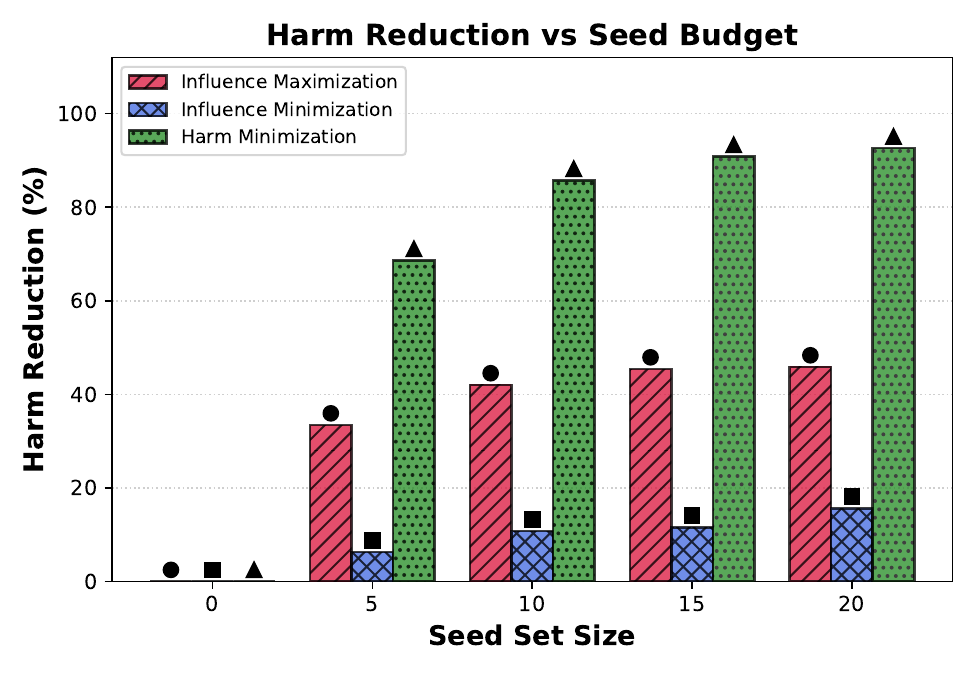}%
    \label{fig:harm_slashdot}
}
\hfill
\subfloat[Wiki-RfA]{%
    \includegraphics[width=0.32\textwidth]{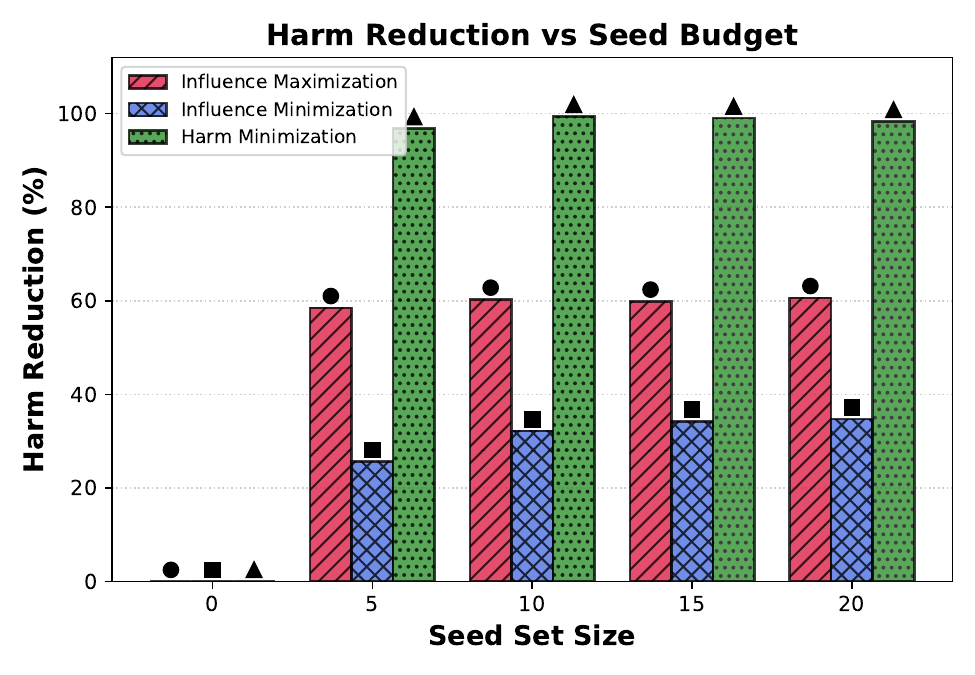}%
    \label{fig:harm_wikirfa}
}

\caption{Percentage harm reduction versus seed set size for Influence
Maximization, Influence Minimization, and Harm Minimization across all
six datasets.}
\label{fig:time_harm}
\end{figure*}
\paragraph{Running time and cost-effectiveness}
Table~\ref{tab:runtime_comparison} situates HM's harm reduction against its computational cost relative to six lighter-weight heuristics. The
structural heuristics (Random, High-Degree, Degree Discount) are
essentially instantaneous but achieve only modest harm reduction
(19--72\% across datasets and budgets); Betweenness and Targeted
Protection incur greater cost, particularly on Epinions and Slashdot,
while improving harm reduction only incrementally. Reactive seeding
remains fast throughout but its harm reduction, while competitive on
smaller networks, falls short of HM on the larger ones. Greedy HM is
the most computationally expensive method at every budget, with
running time scaling with both seed budget and network size,
consistent with the $O(k\,|V|\,T\,|E|)$ complexity bound yet it is also, without exception, the
highest-harm-reducing method in every cell of the table \ref{tab:runtime_comparison}. On Epinions,
for example, HM's running time grows from $33.25$s at $k{=}5$ to
$128.40$s at $k{=}20$, while its harm reduction correspondingly rises
from $46\%$ to $81\%$, a margin no baseline approaches at any tested
budget. This establishes HM's computational overhead as the direct
and interpretable cost of its selectivity, rather than an incidental
inefficiency- the additional time is spent evaluating candidates
against the shortfall-weighted objective that the cheaper heuristics
do not compute at all.
\begin{table*}[t!]
\centering
\caption{Comparison of seed-selection running time (seconds) and harm
reduction (\%) across datasets and seed-set sizes for the Harm
Minimization objective. Each entry reports running time followed by
harm reduction in parentheses.}
\label{tab:runtime_comparison}
\footnotesize
\setlength{\tabcolsep}{5.2pt}
\begin{tabular}{cc|c|c|c|c|c|c|c}
\toprule
\textbf{Dataset} & \textbf{$k$}
& \textbf{Random}
& \textbf{High-Deg.}
& \textbf{Deg. Discount}
& \textbf{Betweenness}
& \textbf{Targeted Prot.}
& \textbf{Reactive}
& \textbf{Greedy HM} \\
\midrule

\multirow{4}{*}{BA}
& 5
& 0.001 (38) & 0.001 (52) & 0.004 (61) & 0.170 (69)
& 0.20 (75) & 0.003 (82) & \textbf{3.05 (90)} \\
& 10
& 0.001 (46) & 0.001 (61) & 0.004 (69) & 0.171 (77)
& 0.34 (83) & 0.003 (91) & \textbf{6.02 (98)} \\
& 15
& 0.001 (48) & 0.001 (63) & 0.004 (71) & 0.172 (78)
& 0.43 (84) & 0.003 (92) & \textbf{8.91 (99)} \\
& 20
& 0.001 (49) & 0.001 (64) & 0.004 (72) & 0.172 (79)
& 0.51 (85) & 0.003 (93) & \textbf{11.84 (100)} \\

\midrule

\multirow{4}{*}{Bitcoin Alpha}
& 5
& 0.001 (41) & 0.001 (55) & 0.006 (65) & 0.31 (74)
& 0.38 (82) & 0.006 (91) & \textbf{5.42 (98)} \\
& 10
& 0.001 (44) & 0.001 (58) & 0.006 (68) & 0.32 (77)
& 0.65 (84) & 0.006 (94) & \textbf{10.61 (100)} \\
& 15
& 0.001 (45) & 0.001 (59) & 0.007 (69) & 0.32 (78)
& 0.91 (85) & 0.006 (95) & \textbf{15.73 (100)} \\
& 20
& 0.001 (45) & 0.001 (59) & 0.007 (69) & 0.32 (78)
& 1.18 (86) & 0.006 (95) & \textbf{20.89 (99)} \\

\midrule

\multirow{4}{*}{Bitcoin OTC}
& 5
& 0.001 (39) & 0.001 (53) & 0.006 (62) & 0.34 (71)
& 0.42 (78) & 0.006 (87) & \textbf{5.89 (94)} \\
& 10
& 0.001 (43) & 0.001 (57) & 0.006 (67) & 0.35 (76)
& 0.71 (83) & 0.006 (92) & \textbf{11.52 (98)} \\
& 15
& 0.001 (45) & 0.001 (59) & 0.007 (69) & 0.35 (78)
& 1.02 (85) & 0.006 (94) & \textbf{17.08 (100)} \\
& 20
& 0.001 (46) & 0.001 (60) & 0.007 (70) & 0.36 (79)
& 1.31 (86) & 0.006 (95) & \textbf{22.56 (100)} \\

\midrule

\multirow{4}{*}{Epinions}
& 5
& 0.002 (19) & 0.002 (25) & 0.012 (29) & 1.84 (34)
& 2.15 (38) & 0.018 (42) & \textbf{33.25 (46)} \\
& 10
& 0.002 (27) & 0.002 (34) & 0.013 (40) & 1.86 (46)
& 4.20 (51) & 0.018 (57) & \textbf{65.05 (63)} \\
& 15
& 0.002 (32) & 0.002 (41) & 0.013 (47) & 1.87 (55)
& 6.24 (61) & 0.018 (68) & \textbf{96.80 (75)} \\
& 20
& 0.002 (35) & 0.002 (44) & 0.014 (51) & 1.89 (59)
& 8.27 (66) & 0.018 (74) & \textbf{128.40 (81)} \\

\midrule

\multirow{4}{*}{Slashdot}
& 5
& 0.001 (25) & 0.001 (32) & 0.009 (38) & 0.92 (44)
& 1.10 (49) & 0.011 (55) & \textbf{19.15 (61)} \\
& 10
& 0.001 (33) & 0.001 (41) & 0.010 (47) & 0.93 (53)
& 2.15 (59) & 0.011 (66) & \textbf{37.48 (72)} \\
& 15
& 0.001 (37) & 0.001 (45) & 0.010 (52) & 0.94 (58)
& 3.18 (64) & 0.011 (71) & \textbf{55.65 (79)} \\
& 20
& 0.001 (40) & 0.001 (48) & 0.011 (56) & 0.95 (62)
& 4.21 (68) & 0.011 (75) & \textbf{73.58 (84)} \\

\midrule

\multirow{4}{*}{Wiki-RfA}
& 5
& 0.001 (29) & 0.001 (37) & 0.008 (43) & 0.61 (49)
& 0.73 (55) & 0.009 (62) & \textbf{13.32 (69)} \\
& 10
& 0.001 (37) & 0.001 (46) & 0.009 (52) & 0.62 (58)
& 1.42 (64) & 0.009 (71) & \textbf{25.98 (77)} \\
& 15
& 0.001 (41) & 0.001 (50) & 0.009 (57) & 0.63 (63)
& 2.10 (69) & 0.009 (76) & \textbf{38.45 (83)} \\
& 20
& 0.001 (44) & 0.001 (54) & 0.010 (61) & 0.64 (67)
& 2.77 (73) & 0.009 (81) & \textbf{50.72 (88)} \\

\bottomrule
\end{tabular}
\end{table*}

\section{Conclusion}
Classical diffusion models and influence objectives share a
limitation in signed social networks - reaching an individual is not
always beneficial, and failing to reach one is not equally costly
for everyone. Independent Cascade, Linear Threshold, and their
signed extensions model activation as discrete and irreversible,
preventing awareness from weakening or reversing under competing
influence, while Influence Maximization and Influence Minimization
optimize aggregate spread without accounting for individual
vulnerability. This paper addressed both limitations jointly through
RASH, which models awareness as continuous, bounded, and
non-monotonic while remaining monotone and $\gamma$-weakly
submodular in the pure-positive and pure-negative regimes and 
Harm Minimization (HM) objective, which optimizes vulnerability-weighted
awareness shortfall rather than raw spread and, though NP-hard,
inherits the same monotonicity and weak-submodularity structure to
yield a bounded greedy approximation guarantee. Experiments across
six structurally diverse signed networks showed RASH to be the only
evaluated model producing nonzero awareness reversal or suppression,
and HM to achieve the highest harm reduction across all datasets and
budgets, outperforming IM, Inf-Min, and six practical heuristics.
HM's intermediate raw activation reflects its objective of reducing
weighted shortfall rather than maximizing reach, and the widening
trade-off on larger networks follows the $(\gamma^+)^T$-decay
predicted by our theory. RASH and HM shift protective
information campaigns from maximizing reach to minimizing harm
experienced by vulnerable individuals.

\textbf{Limitations.} The greedy algorithm does not eliminate
worst-case seed-set instability permitted by weak submodularity,
observed occasionally as elevated variance rather than systematic
failure. The candidate-pool restriction used for large networks
also makes reported solutions optimal only within a degree-ranked
subset; incorporating high-vulnerability, low-degree individuals
into this pool is a natural extension. Finally, $b_i$ and $w_i$ are
currently derived from structural signals alone, though HM can
directly incorporate richer risk information when available.

\textbf{Future work.} Extending RASH and HM to dynamic signed
networks, where $E^+$, $E^-$, and vulnerability profiles evolve
during a campaign, would support applications such as public-health
advisories and crisis communication. Studying the policy parameter
$p$ as $p \to \infty$, which shifts $H_p(A)$ from broad harm
reduction toward protecting the single worst-off individual, would
further clarify how HM balances aggregate and worst-case
protection.

\bibliographystyle{IEEEtran}
\bibliography{ref}

\end{document}